\documentclass{lmcs}
\pdfoutput=1

\usepackage{lastpage}
\lmcsdoi{15}{4}{10}
\lmcsheading{}{\pageref{LastPage}}{}{}%
{Nov.~02,~2018}{Nov.~27,~2019}{}

\usepackage{amsmath}
\usepackage{amsfonts}
\usepackage{amssymb}
\usepackage{amsthm}
\usepackage{ragged2e}
\usepackage[ruled, vlined]{algorithm2e}
\usepackage{tikz}
\usetikzlibrary{arrows,shapes,decorations,automata,backgrounds,petri,matrix,decorations.pathmorphing,backgrounds,fit}
\usetikzlibrary{positioning}
\usepackage{tikz-qtree}
\usetikzlibrary{trees}
\usepackage{framed}

\newcommand{\N}{\mathbb{N}}
\newcommand{\fun}[3]{\ensuremath{#1\colon #2 \to #3}}

\newcommand{\A}{\mathcal{A}}
\newcommand{\B}{\mathcal{B}}
\newcommand{\C}{\mathcal{C}}
\newcommand{\D}{\mathcal{D}}
\newcommand{\w}{\mathrm{width}}
\newcommand{\rk}{\mathrm{rk}}
\newcommand{\Adet}{\mathcal{A}_{\mathrm{det}}}
\newcommand{\Atriv}{\mathcal{A}_{\mathrm{triv}}}
\newcommand{\CC}{\mathcal{C}}
\newcommand{\trans}[1]{\stackrel{#1}{\longrightarrow}}

\newcommand{\srelT}[2]{
  \xrightarrow
  [\protect{\raisebox{1mm}[0pt][0pt]{\ensuremath{\scriptstyle #2}}}]
  {#1}
  {\!\!}^*\,
}

\newcommand{\Gg}{\mathcal{G}}
\newcommand{\GG}{\mathcal{G}}
\newcommand{\Gk}{\GG_k}

\newcommand{\inc}{\sqsubseteq}
\newcommand{\inck}{\sqsubseteq_k}

\newcommand{\cl}{\mathit{cl}}

\definecolor{newgreen}{RGB}{10,130,20}

\begin{document}

\title{Computing the Width of Non-deterministic Automata}

\author[D.~Kuperberg]{Denis Kuperberg}
\address{CNRS, LIP, ENS Lyon}
\email{denis.kuperberg@ens-lyon.fr}
\thanks{Work supported by the Grant Palse Impulsion}

\author[A.~Majumdar]{Anirban Majumdar}
\address{LSV, CNRS \&  ENS Paris-Saclay, Univ. Paris-Saclay --
  Cachan (France) \newline Univ. Rennes, Inria, CNRS, IRISA -- Rennes
  (France)} 
\email{majumdar@lsv.fr}


\begin{abstract}
We introduce a measure called width, quantifying the amount of nondeterminism in automata. Width generalises the notion of good-for-games (GFG) automata, that correspond to NFAs of width $1$, and where an accepting run can be built on-the-fly on any accepted input. We describe an incremental determinisation construction on NFAs, which can be more efficient than the full powerset determinisation, depending on the width of the input NFA. This construction can be generalised to infinite words, and is particularly well-suited to coB\"uchi automata. For coB\"uchi automata, this procedure can be used to compute either a deterministic automaton or a GFG one, and it is algorithmically more efficient in the last case. We show this fact by proving that checking whether a coB\"uchi automaton is determinisable by pruning is NP-complete. On finite or infinite words, we show that computing the width of an automaton is EXPTIME-complete. This implies EXPTIME-completeness for multipebble simulation games on NFAs.
\end{abstract}

\maketitle

\section{Introduction}

Determinisation of non-deterministic automata (NFAs) is one of the cornerstone problems of automata theory. Determinisation algorithms occupy a central place in the theoretical study of regular languages of finite or infinite words, inducing for instance many of the robustness properties of these classes. Moreover, determinisation algorithms are not only used to prove theoretical properties related with decidability and complexity, but are also used when we want to put these theories to practical use, with many applications for instance in verification and synthesis.
Consequently, there is a very active field of research aiming at optimizing or approximating determinisation, or circumventing it in contexts like inclusion of NFA or Church Synthesis. Indeed, determinisation is a costly operation, as the state space blow-up is in $O(2^n)$ on finite words, $O(3^n)$ for coB\"uchi automata \cite{MH84}, and $2^{O(n\log(n))}$ for B\"uchi automata \cite{Safra88}, where there is also an increased complexity of the acceptance condition, going from B\"uchi to Rabin.

If $\A$ and $\B$ are NFAs, the classical way of checking the inclusion $L(\A)\subseteq L(\B)$ is to determinise $\B$, complement it, and test emptiness of $L(A)\cap \overline{L(B)}$. To circumvent a full determinisation, the recent algorithm from \cite{BP13} proved to be very efficient, as it is likely to explore only a part of the powerset construction.
Other approaches use simulation games to approximate inclusion at a cheaper cost, see for instance \cite{Ete02}.

%

Another way of avoiding a full determinisation construction consists in replacing determinism by a weaker constraint that suffices in some particular contexts.
In this spirit, Good-for-Games (GFG for short) automata were introduced in \cite{HP06}, as a way to solve the Church synthesis problem. 
This problem \cite{Church63} asks, given a specification $L$, typically given by an LTL formula over an alphabet of inputs and outputs, whether there is a reactive system (transducer) whose behaviour is included in $L$. The classical solution \cite{Rabin72} computes a deterministic automaton for $L$, and solves a game defined on this automaton. It turns out that replacing determinism by the weaker constraint of being GFG is sufficient in this context.
Intuitively, a GFG automaton is a non-deterministic automaton where it is possible to build an accepting run in an online way, without any knowledge of the future, provided the input word is in the language of the automaton. In \cite{HP06}, it is shown that GFG automata allow an incremental algorithm for the Church synthesis problem: we can build increasingly large games, with the possibility that the algorithm stops before the full determinisation is needed. 
One of the aims of this paper is to generalise this idea to determinisation of NFA, for use in any context and not only Church synthesis. We give an incremental determinisation construction, where the emphasis is on space-saving, and that allows in some cases to avoid the full powerset construction.
\medskip

The notion of width introduced in \cite{KM18} (of which this paper is an extended version) generalises the GFG model, by allowing more than one run to be built in an online way. Intuitively, width quantifies how many states we have to keep track of simultaneously in order to build an accepting run in an online way. The maximal width of an automaton is its number of states. The width of an automaton also corresponds to the number of steps performed by our incremental determinisation construction before stopping. In the worst case where the width is equal to the number of states of the automaton, we end up performing the full powerset construction (or its generalisations for infinite words). We study here the complexity of directly computing the width of a nondeterministic automaton, and we show that it is EXPTIME-complete, even in the restricted case of universal safety automata. This constitutes a new contribution compared to the conference version of this paper \cite{KM18}, where only PSPACE-hardness was shown for the width problem.

We obtain this result via a reduction from a combinatorial game on boolean formulas from \cite{SC79}.
In the process, we also show that multipebble simulation games on NFAs are EXPTIME-complete, even when testing simulation of a trivial automaton by an NFA of size $n$, using a fixed number of $n/2$ pebbles.
This generalizes a previous result from \cite{Clemente12}, where EXPTIME-completeness is shown for multipebble simulations on B\"uchi automata, with a number of pebbles fixed to $\sqrt{n}$.

\medskip

The properties of GFG automata and their links with other models and algorithms (tree automata, Markov Decision Processes, efficient algorithms for parity games) are nowadays actively investigated \cite{BKKS13,KMBK14,KS15,BKS17,BK18,AK19,IK19,CF19}. Colcombet introduced a generalisation of the concept of GFG called history-determinism \cite{Col09}, replacing determinism in the framework of automata with counters. It was conjectured by Colcombet \cite{Col12} that GFG automata were essentially deterministic automata with additional useless transitions. It was shown in \cite{KS15} that on the contrary there is in general an exponential state space blowup to translate GFG automata to deterministic ones. GFG automata retain several good properties of determinism, in particular they can be composed with trees and games, and easily checked for inclusion.  

We give here the first algorithms allowing to build GFG automata from arbitrary non-deterministic automata on infinite words, allowing to potentially save exponential space compared to deterministic automata. Our incremental constructions look for small GFG automata, and aim at avoiding the worst-case complexities of determinisation constructions. Moreover, in the case of coB\"uchi automata, we show that the procedure is more efficient than its analog looking for a deterministic automaton, since checking the GFG property is polynomial \cite{KS15}, while we show here that the corresponding step for determinisation, that is checking whether a coB\"uchi automaton is Determinisable By Pruning (DBP) is NP-complete. Combined with the good properties of GFG coB\"uchi automata related to succinctness even for LTL-definable languages \cite{IK19} and polynomial time minimization \cite{AK19}, this makes the class of coB\"uchi automata especially well-suited for this approach.
\medskip

%

As a measure of non-determinism, width can be compared with ambiguity, where the idea is to limit the number of possible runs of the automaton. In this context unambiguous automata play a role analogous to GFG automata for width.
Unambiguous automata are studied in \cite{Leung06}, degrees of ambiguity are investigated in \cite{WS91,Leung98, Leung05}. We give examples of automata with various width and ambiguity, showing that these two measures are essentially orthogonal.
\medskip

After defining automata and games in Section \ref{sec:def}, we describe the width approach on finite words and the incremental determinisation construction for NFAs in Section \ref{sec:NFA}. We compare width and ambiguity in Section \ref{sec:ambig}, and show a link between width and multipebble simulation relations in Section \ref{sec:sim}. We show in Section \ref{sec:exptime} that computing the width of a NFA, as well as testing whether a multipebble simulation holds, is EXPTIME-complete. We then move to infinite words, and start by focusing on the coB\"uchi acceptance condition in Section \ref{sec:cobuchi}. We show that the breakpoint construction \cite{MH84} can be adapted to yield an incremental breakpoint construction, that can be used to build either a deterministic or a GFG coB\"uchi automaton from a nondeterministic one. We compare the two approaches, and exhibit several advantages of GFG automata in this special case. We finally describe the general case of B\"uchi automata in Section \ref{sec:buchi}, where we give an incremental version of the Safra construction \cite{Safra88}, and point to open problems related to the algorithmic complexity of this approach.

\section{Definitions}\label{sec:def}

We will use $\Sigma$ to denote a finite alphabet. The empty word is denoted $\varepsilon$.
If $i\leq j$, the set $\{i,i+1,i+2,\dots,j\}$ is denoted $[i,j]$.
If $X$ is a set and $k\in\N$, we note $X^{\leq k}$ the set of subsets of $X$ of size at most $k$.
The complement of a set $X$ is denoted $\overline{X}$.
If $u\in \Sigma^*$ is a word and $L\subseteq \Sigma^*$ is a language, the left quotient of $L$ by $u$ is $u^{-1}L:=\{v\in\Sigma^*\mid uv\in L\}$.

\subsection{Automata}


A non-deterministic automaton $\A$ is a tuple $(Q,\Sigma,q_0,\Delta,F)$ where $Q$ is the set of states, $\Sigma$ is a finite alphabet, $q_0\in Q$ is the initial state, $\Delta: Q\times\Sigma\to 2^Q$ is the transition function, and $F\subseteq Q$ is the set of accepting states.

The transition function is naturally generalised to $2^Q$ by setting for any $(X,a)\in 2^Q\times\Sigma$, $\Delta(X,a)$ the set of $a$-successors of $X$, i.e. $\Delta(X,a)=\{q\in Q\mid \exists p\in X, q\in\Delta(p,a)\}$.

We will sometimes identify $\Delta$ with its graph, and write $(p,a,q)\in\Delta$ instead of $q\in\Delta(p,a)$.

If for all $(p,a)\in Q\times\Sigma$ there is a unique $q\in Q$ such that $(p,a,q)\in\Delta$, we say that $\A$ is \emph{deterministic}.

If $u=a_1\dots a_n$ is a finite word of $\Sigma^*$, a run of $\A$ on $u$ is a sequence $q_0q_1\dots q_n$ such that for all $i\in[1,n]$, we have $q_i\in\Delta(q_{i-1},a_i)$. The run is said to be \emph{accepting} if $q_n\in F$.

If $u=a_1a_2\dots$ is an infinite word of $\Sigma^\omega$, a run of $\A$ on $u$ is a sequence $q_0q_1q_2\dots$ such that for all $i>0$, we have $q_i\in\Delta(q_{i-1},a_i)$. A run is said to be \emph{B\"uchi accepting} if it contains infinitely many accepting states, and \emph{coB\"uchi accepting} if it contains finitely many non-accepting states. Automata on infinite words will be called B\"uchi and coB\"uchi automata, to specify their acceptance condition.

We will note NFA (resp. DFA) for a non-deterministic (resp. deterministic) automaton on finite words, NBA (resp. DBA) for a non-deterministic (resp. deterministic) B\"uchi automaton, and NCA (resp. DCA) for a non-deterministic (resp. deterministic) coB\"uchi automaton.

We also mention more general acceptance conditions on infinite words:
\begin{itemize}
\item \emph{Parity condition}: each state $q$ has a rank $\rk(q)\in\N$, and an infinite run is accepting if the highest rank appearing infinitely often is even.
\item \emph{Rabin condition}: a set $\{(G_1,B_1),\dots,(G_k,B_k)\}$ of pairs with $G_i,B_i\subseteq Q$ is given, a run is accepting if there exists $i\in[1,k]$ such that the run contains infinitely many states from $G_i$ and finitely many states from $B_i$.
\item \emph{Streett condition}: dual of the Rabin condition.
\item \emph{Muller condition}: a set $\{F_1,\dots, F_k\}$ of subsets of $Q$ is given, a run is accepting if there is $i\in[1,k]$ such that the set of states appearing infinitely often is exactly $F_i$.
\end{itemize}
\noindent
The language of an automaton $\A$, noted $L(\A)$, is the set of words on which the automaton $\A$ has an accepting run. Two automata are said \emph{equivalent} if they recognise the same language. An automaton $\A$ is said \emph{universal} if it accepts all words.

An automaton $\A$ is \emph{determinisable by pruning} (DBP) if an equivalent deterministic automaton can be obtained from $\A$ by removing some transitions.

An automaton $\A$ is \emph{Good-For-Games} (GFG) if there exists a function $\fun{\sigma}{\Sigma^\ast}{Q}$ (called \emph{GFG strategy}) that resolves the non-determinism of $\A$ depending only on the prefix of the input word read so far: over every word $u=a_1a_2a_3\dots$ (finite or infinite depending on the type of automaton considered), the sequence of states $\sigma(\varepsilon)\sigma(a_1)\sigma(a_1a_2)\sigma(a_1a_2a_3)\dots$ is a run of $\A$ on $u$, and it is accepting whenever $u\in L(\A)$. For instance every DBP automaton is GFG. See \cite{BKKS13} for more introductory material and examples on GFG automata.

\subsection{Games}
%

A \emph{game} $\Gg=(V_0,V_1,v_I,E, W)$ of infinite duration between two players $0$ and $1$ consists of: a finite set of \emph{positions} $V$ being a disjoint union of $V_0$ and $V_1$; an \emph{initial position} $v_I\in V$; a set of \emph{edges} $E\subseteq V\times V$; and a \emph{winning condition} $W\subseteq V^\omega$.

A \emph{play} is an infinite sequence of positions $v_0v_1v_2\dots\in V^\omega$ such that $v_0=v_I$ and for all $n\in \mathbb N$, $(v_n,v_{n+1})\in E$. A play $\pi\in V^\omega$ is \emph{winning} for Player 0 if it belongs to $W$. Otherwise $\pi$ is \emph{winning} for Player 1.

A \emph{strategy} for Player 0 (resp. 1) is a function $\fun{\sigma_0}{V^\ast\times V_0}{V}$ (resp. $\fun{\sigma_1}{V^\ast\times V_1}{V}$), describing which edge should be played given the history of the play $u\in V^\ast$ and the current position $v\in V$. A strategy $\sigma_P$ has to obey the edge relation, i.e.~there has to be an edge in $E$ from $v$ to $\sigma_P(u,v)$. A play $\pi=v_0v_1v_2\dots$ is \emph{consistent} with a strategy $\sigma_P$ of a player $P$ if for every $n$ such that $v_n\in V_P$ we have $v_{n+1}=\sigma_P(v_0\ldots v_{n-1},v_n)$.

A strategy for Player 0 (resp. Player 1) is \emph{positional} if it does not use the history of the play, i.e. it can be seen as a function $V_0\to V$ (resp. $V_1\to V$).

We say that a strategy $\sigma_P$ of a player $P$ is \emph{winning} if every play consistent with $\sigma_P$ is winning for $P$. In this case, we say that $P$ \emph{wins} the game $\Gg$.

A game is \emph{positionally determined} if exactly one of the players has a positional winning strategy in the game.

\section{Finite words}\label{sec:NFA}

\subsection{Width of an NFA}

\label{sec:width}

%

We want to define the \emph{width} of an NFA as the minimum number of simultaneous states that need to be tracked in order to be able to deterministically build an accepting run in an online way. 
Let $\A=(Q,\Sigma,q_0,\Delta,F)$ be an NFA, and $n=|Q|$ be the size of $\A$.
In order to define this notion of width formally, we introduce a family of games $\Gg_w(\A,k)$, parameterized by an integer $k\in[1,n]$.

\newcommand{\Qk}{Q^{\leq k}}
\newcommand{\Gw}{\Gg_w}

The game $\Gw(\A,k)$ is played on $\Qk$, starts in $X_0=\{q_0\}$, and the round $i$ of the game from a position $X_i\in \Qk$ is defined as follows:
\begin{itemize}
\item Player 1 chooses a letter $a_{i+1}\in \Sigma$.
\item Player 0 moves to a subset $X_{i+1}\subseteq \Delta(X_i,a_{i+1})$ of size at most $k$.
\end{itemize}
\noindent
A play is winning for Player $0$ if for all $r\in\mathbb N$, whenever $a_1a_2\dots a_r\in L(\A)$, $X_r$ contains an accepting state.

\begin{defi}
The width of an NFA $\A$, denoted $\w(\A)$, is the least $k$ such that Player $0$ wins $\Gw(\A,k)$.
\end{defi}

Intuitively, the width measures the ``amount of non-determinism'' in an automaton: it counts the number of simultaneous states we have to keep track of, in order to be sure to find an accepting run in an online way.

\begin{fact}
An NFA $\A$ is GFG if and only if $\w(\A)=1$.
\end{fact}

\subsection{Partial powerset construction}
\label{sec:partial}
\newcommand{\Ak}{\A_k}

We give here a generalisation of the powerset construction, following the intuition of the width measure.

We define the $k$-subset construction of $\A$ to be the subset construction where the size of each set is bounded by $k$. Formally, it is the NFA $\Ak = (\Qk,\Sigma,\{q_0\},\Delta',F')$ where:
\begin{itemize}
\item $\Delta'(X,a) := 
\begin{cases}
    \{\Delta(X,a)\}						   & \text{if }|\Delta(X,a)| \leq k\\
    \{X'\mid X' \subseteq \Delta(X,a), |X'|=k\}              & \text{otherwise}
\end{cases}$
\item $F' := \{X\in\Qk \mid X\cap F \neq \emptyset\}$
\end{itemize}

\begin{rem}\label{rem:A1}
Notice that $\A_1$ is isomorphic to the automaton $\A$.
\end{rem}

\begin{lem}\label{lem:sizek}
The automaton $\Ak$ has less than $\dfrac{n^k}{(k-1)!}+1$ states.
\end{lem}
\begin{proof}
The number of states of $\Ak$ is (at most) $|\Qk|=\sum_{i=0}^k \binom{n}{i}$. Using the fact that $\binom{n}{i}\leq \frac{n^i}{i!}$, we can bound the number of states of $\Ak$ by $\sum_{i=0}^k \frac{n^i}{i!}\leq \sum_{i=0}^k \frac{n^k}{k!}\leq 1+\sum_{i=1}^k \frac{n^k}{k!}=\frac{n^k}{(k-1)!}+1$.
\end{proof}

The following lemma shows the link between width and the $k$-powerset construction.
\begin{lem}\label{lem:AkGFG}
One has $\w(\A)\leq k$ if and only if $\Ak$ is GFG.
\end{lem}

\begin{proof}
Winning strategies in $\Gw(\A,k)$ are in bijection with GFG strategies for $\Ak$.
\end{proof}

\subsection{GFG automata on finite words}
We recall here results on GFG automata on finite words. 

We start with a lemma characterizing GFG strategies. Let $\A=(Q,\Sigma,q_0,\Delta,F)$ be an NFA recognising a language $L$, and $\sigma:\Sigma^*\to Q$ be a potential GFG strategy. For any $q\in Q$, we denote $L(q)$ the language accepted from $q$ in $\A$, i.e. $L(q)$ is the language of $\A$ with $q$ as initial state.

\begin{lem}\label{lem:quotient}
$\sigma$ is a GFG strategy if and only if for all $u\in\Sigma^*$, $L(\sigma(u))=u^{-1}L$.
\end{lem}

\begin{proof}
%
Assume $\sigma$ is a GFG strategy, and let $u\in\Sigma^*$. Let $q=\sigma(u)$. It is clear that $L(q)\subseteq u^{-1}L$, as any run accepting $v$ from $q$ is a witness that $uv\in L$ (together with the run on $u$ reaching $q$ from $q_0$). We therefore have to show that for all $v\in u^{-1}L$, we have $v\in L(q)$. For this, recall that $\sigma$ is a GFG strategy, so $\sigma(uv)\in F$. Since $\sigma(u)=q$, there is an accepting run starting in $q$ and labelled by $v$, showing $v\in L(q)$.

Conversely, assume that for any $u\in\Sigma^*$, $L(\sigma(u))=u^{-1}L$. In particular, it means, that for any $u\in L$ we have $\varepsilon\in L(\sigma(u))$, so $\sigma(u)$ is an accepting state. This implies that $\sigma$ is indeed a GFG strategy.
\end{proof}

We now go to the main result of this section. This result on DBP automata has first been proved in \cite{AKL09}, and then a more general version allowing lookahead was proved using a game-based approach in \cite{LR13}. The link between GFG and DBP automata on finite words was first mentioned in \cite{Col12}.

\begin{thm}\label{thm:loding} \cite{AKL09,Col12,LR13}
An NFA $\A$ is GFG if and only if it is DBP. Moreover, there is a quadratic algorithm that determines whether an NFA is GFG, and in the positive cases builds an equivalent DFA by removing transitions.
\end{thm}

\subsection{Incremental determinisation procedure}
\label{sec:progNFA}

We can now describe an incremental determinisation procedure, aiming at saving resources in the search of a deterministic automaton. In the process, we also compute the width of the input NFA.

The algorithm goes as follows:

\begin{algorithm}[h]
\KwIn{NFA $\A$}
\KwOut{$\w(\A)$ and DFA $\D$ equivalent to $\A$}
 $k:=1$\;
\While{$\Ak$ is not GFG}{
$k:=k+1$\;
Construct $\Ak$\;
} Compute an equivalent DFA $\D$ from $\Ak$ by removing transitions\;
 Return $k,\D$\;
    \caption{{\bf Incremental NFA determinisation} \label{alg:incrNFA}}
\end{algorithm}
The usual determinisation procedure uses the full powerset construction, i.e. assumes that we are in the case of maximal width. Once a deterministic automaton has been obtained, be it by full determinization or via our incremental approach, it can be minimized easily.

Our method here is to approach this powerset construction ``from below'', and incrementally increase the width until we find the good one. In some cases, this allows to compute directly a smaller automaton, and avoids using the full powerset construction of exponential state complexity as an intermediary step.

For an NFA with $n$ states and width $k$, the complexity of this algorithm is in $O\big(\frac{n^{2k}}{(k-1)!^2}\big)$, by Lemma \ref{lem:sizek} and Theorem \ref{thm:loding}.
%
\begin{exa}\label{ex:width2}
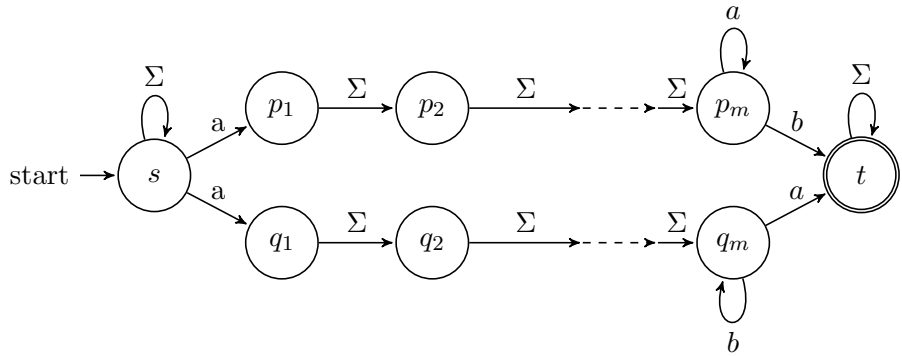
\begin{figure}[ht]
\centering
\begin{tikzpicture}[->,>=stealth',shorten >=1pt,auto,node distance=2cm,
                    semithick]
  \tikzstyle{every state}=[fill=none,draw=black,text=black]

  \node[initial,state] (A)                    {$s$};
  \node[state]         (B1) [above right= .2cm and 1cm of A] 	  {$p_1$};
  \node[state]         (C1) [right of=B1] 	  {$p_2$};

  \coordinate[right of=C1] (d1);
  \coordinate[right=1cm of d1] (d1');
  \node[state]         (D1) [right of=d1] 	  {$p_m$};

\node[state]         (B2) [below right= .2cm and 1cm of A] 	  {$q_1$};
  \node[state]         (C2) [right of=B2] 	  {$q_2$};

  \coordinate[right of=C2] (d2);
  \coordinate[right=1cm of d2] (d2');
  \node[state]         (D2) [right of=d2] 	  {$q_m$};

  \node[state,accepting]         (A2) [above right=.2cm and 1 cm of D2] 	  {$t$};

 \path (A) edge 	node [above] {a} (B1)
	 (B1) edge 	node [above] {$\Sigma$} (C1)
	 (C1) edge 	node[above] {$\Sigma$} (d1)
	  (d1) edge[dashed]  (d1')
	 (d1') edge	node [above] {$\Sigma$} (D1)
	 (D1) edge 	node [above] {$b$} (A2)

	(A) edge 	node [above] {a} (B2)
	 (B2) edge 	node [above] {$\Sigma$} (C2)
	 (C2) edge 	node[above] {$\Sigma$} (d2)
	 (d2) edge[dashed]  (d2')
	 (d2') edge	node [above] {$\Sigma$} (D2)
	 (D2) edge 	node [above] {$a$} (A2)

	(A) edge [loop above] node {$\Sigma$} (A)
	(D1) edge [loop above] node {$a$} (D1)
	(D2) edge [loop below] node {$b$} (D2)
	(A2) edge [loop above] node {$\Sigma$} (A2)
;
\end{tikzpicture}
\caption{Example: 2-subset construction is enough}\label{fig:M4}
\end{figure}
We take an example in Figure \ref{fig:M4}. Here the language recognised by this automaton is $L(\A) = \Sigma^*a\Sigma^{\geq m}$, and it has width 2. Indeed, the automaton $\A_2$ is DBP, and can be pruned to keep only states $\{s\}, \{p_1,q_1\},\dots\{p_m,q_m\},\{t\}$ (so getting rid of states such as $\{s,p_1\}$) while still recognizing $L(\A)$. Therefore, our determinisation procedure uses time $O(n^4)$ and builds an intermediary DFA $\A_2$ of size $O(n^2)$, while a classical determinisation via powerset construction would build an exponential-size DFA. The pruning process of the DBP automaton $\A_2$ yields here the minimal DFA of size $m+2=O(n)$.
\end{exa}

\begin{figure}[ht]
\centering
\begin{tikzpicture}[->,>=stealth',shorten >=1pt,auto,node distance=2cm,
                    semithick]
  \tikzstyle{every state}=[fill=none,draw=black,text=black,inner sep = 0pt,minimum size=.7cm]

  \node[initial,state] (A)                    {$s$};

  \node[state]         (C1) [above right =.5cm and 2cm of A] 	  {$p_2$};
  \node[state]         (B1) [above =1cm of C1] 	  {$p_1$};
  \node        (D1) [below =1cm of C1] 	  {$\vdots$};
  \node[state]         (E1) [below =1cm of D1] 	  {$p_n$};

%
%

  \node[state,accepting]         (A2) [below right=.5cm and 2cm of C1] 	  {$t$};

 \path (A) edge 	node [above] {$\Sigma$} (B1)
	 (A) edge 	node [above] {$\Sigma$} (C1)
	 (A) edge[dotted]  (D1)
	 (A) edge 	node[above] {$\Sigma$} (E1)
%
%
	(B1) edge  node[above=.1cm] {$a_1$} (A2)
	(C1) edge  node[above] {$a_2$} (A2)
	(D1) edge[dotted]  (A2)
	(E1) edge node[above=.07cm] {$a_n$} (A2)
;
\end{tikzpicture}
\caption{Example: Subset construction can be efficient}\label{fig:Subset_const_isbetter}
\end{figure}
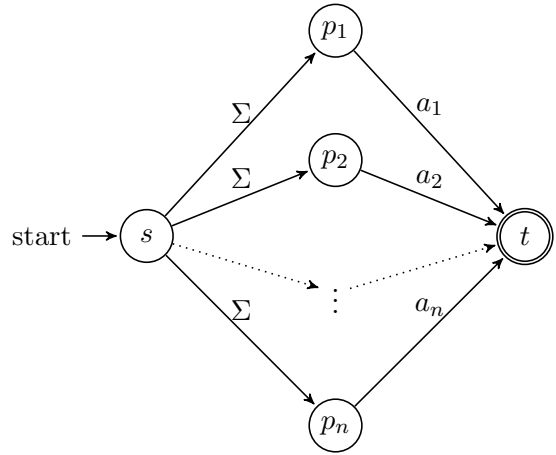

But in some other cases, the powerset construction is actually more efficient than the $k$-powerset construction, in terms of number of reachable states. For instance consider an example where the alphabet is $\Sigma=\{a_1,a_2,\dots,a_n\}$ and the automaton has $n+2$ states: one initial state, one final state and $n$ other transition states, as shown in Figure \ref{fig:Subset_const_isbetter}. The transition relation is defined as in the picture.
On this example, the automaton obtained from subset construction has only $3$ states whereas for any $k$, the automaton obtained by $k$-subset construction will have $\binom{n}{k}+2$ states. 
This example illustrates that sometimes the powerset construction can actually be more efficient than the $k$-powerset construction, and the incremental $k$-powerset construction is not necessarily increasing in terms of number of states as $k$ grows.  It would therefore be interesting to be able to either run the two methods in parallel, or guess which one is more effiicent based on the shape of the input NFA.

\section{Width versus Ambiguity}\label{sec:ambig}


In this section, we recall a useful notion of automata, namely \emph{ambiguity} from \cite{Leung98} and investigate its relation with the notion of \emph{width} in the form of examples.

\newcommand{\amb}{\mathit{amb}}

\begin{defi}
Given an NFA $\A$ and a word $w$, the ambiguity of $w$ is the number of different accepting paths for $w$ in $\A$.
\end{defi}

Note that a word is accepted by $\A$ if and only if the ambiguity of the word is non-zero. $\A$ is called \emph{unambiguous} if ambiguity of any word is either zero or one. $\A$ is called \emph{finitely} (resp. \emph{polynomially, exponentially}) ambiguous if there exists a constant (resp. polynomial, exponential) function $f$  such that the ambiguity of any word of length $n$ is bounded by $f(n)$.
%

For example, every DFA is unambiguous since every word accepted by a DFA has a unique accepting run. Some more illustrated examples are given in this section, showing that width and ambiguity can vary independently from each other in NFAs.

\subsection{Width 1, Exponentially ambiguous}
Consider the following NFA accepting all words in $\Sigma^*$.\\

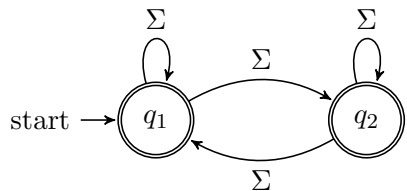
\centering
\begin{tikzpicture}[->,>=stealth',shorten >=1pt,auto,node distance=2.8cm,
                    semithick]
  \tikzstyle{every state}=[fill=none,draw=black,text=black]

  \node[initial,state,accepting] (A)                    {$q_1$};
  \node[state,accepting]         (B) [right of=A] 	  {$q_2$};
\path 	(A) edge [loop above] node {$\Sigma$} (A)
	(B) edge [loop above] node {$\Sigma$} (B)
	(A) edge [bend left]	node [above] {$\Sigma$} (B)
	(B) edge [bend left]	node [below] {$\Sigma$} (A);
\end{tikzpicture}
\justify
The above automaton is exponentially ambiguous but not polynomially ambiguous. Indeed each word of length $n$ has $2^n$ accepting runs. However it has width $1$, since it suffices to stay in $q_1$ to produce an accepting run.

\subsection{Width n, Unambiguous}
Consider the following NFA $\A_n$, recognizing the language $\Sigma^*0\Sigma^{n-1}$.

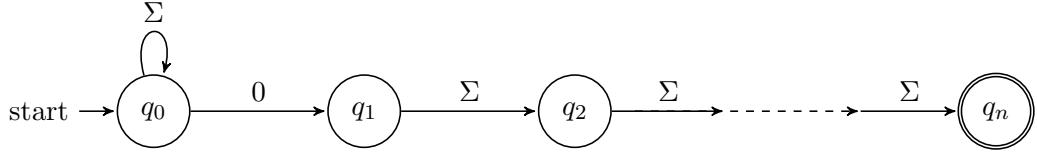
\centering
\begin{tikzpicture}[->,>=stealth',shorten >=1pt,auto,node distance=2.8cm,
                    semithick]
  \tikzstyle{every state}=[fill=none,draw=black,text=black]

  \node[initial,state] (A)                    {$q_0$};
  \node[state]         (B) [right of=A] 	  {$q_1$};
  \node[state]         (C) [right of=B] 	  {$q_2$};

  \coordinate[right of=C] (d1);
  \node[state,accepting]         (D) [right of=d1] 	  {$q_n$};

 \path (A) edge 	node [above] {0} (B)
	 (B) edge 	node [above] {$\Sigma$} (C)
	 (C) edge 	node[above] {$\Sigma$}++(2,0)
	 (d1)++(-2,0) edge[dashed] 	node [above] {} ++(3,0)
	 (d1)++(1,0) edge	node [above] {$\Sigma$} (D)
	(A) edge [loop above] node {$\Sigma$} (A)
;
\end{tikzpicture}
\justify Every word which is in the language of this automaton is accepted by a unique run of $\A_n$. Therefore, it is an unambiguous automaton.\\ But one can show that the minimal DFA for $\A_n$ has exactly $2^n$ states. By Theorem \ref{thm:loding}, this implies that this automaton has width $n$. Indeed, if the width was $k<n$, we could build a deterministic automaton with strictly less than $2^n$ states.\\
More precisely, this automaton has $n+1$ states and width $n$, so this is an example of an NFA $\A_n$ of width $|\A_n|-1$.

\subsection{Width n, Exponentially ambiguous}

Consider the following NFA $\A_n$ , recognizing to the language $L_n = (0+(01^*)^{n-1}0)^*$.\\

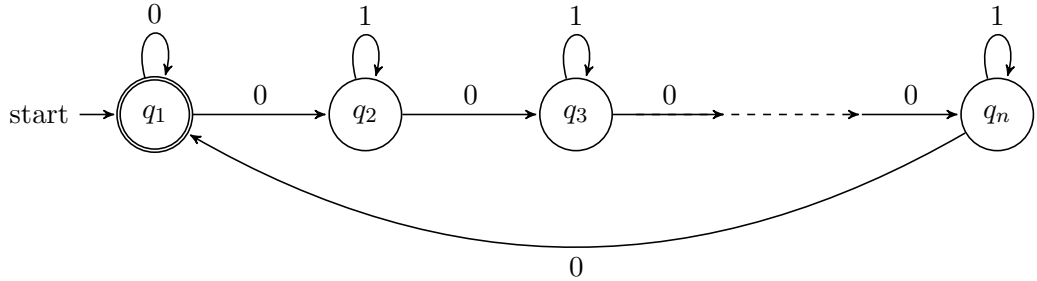
\centering
\begin{tikzpicture}[->,>=stealth',shorten >=1pt,auto,node distance=2.8cm,
                    semithick]
  \tikzstyle{every state}=[fill=none,draw=black,text=black]

  \node[initial,state,accepting] (A)                    {$q_1$};
  \node[state]         (B) [right of=A] 	  {$q_2$};
  \node[state]         (C) [right of=B] 	  {$q_3$};

  \coordinate[right of=C] (d1);
  \node[state]         (D) [right of=d1] 	  {$q_n$};

 \path (A) edge 	node [above] {0} (B)
	 (B) edge 	node [above] {0} (C)
	 (C) edge 	node[above] {0}++(2,0)
	 (d1)++(-2,0) edge[dashed] 	node [above] {} ++(3,0)
	 (d1)++(1,0) edge	node [above] {0} (D)
	 (D) edge[bend left] 	node [below] {0} (A)
	(A) edge [loop above] node {0} (A)
	(B) edge [loop above] node {1} (B)
	(C) edge [loop above] node {1} (C)
	(D) edge [loop above] node {1} (D)
;
\end{tikzpicture}
\justify It is shown in \cite{Leung98} that $\A_n$ is exponentially ambiguous but not polynomially ambiguous, and that any DFA (actually any polyniomally ambiguous NFA) recognising $L_n$ must have $2^n-1$ states. Therefore, $\A_n$ has width $n$ by Theorem \ref{thm:loding}, as in the previous example.

\section{Relation with multipebble simulations}\label{sec:sim}

\subsection{Multipebble Simulation}

Let $\A=(Q_\A,\Sigma,q_\A^0,\Delta_\A,F_\A)$ and $\B=(Q_\B,\Sigma,q_\B^0,\Delta_\B,F_\B)$ be NFAs, and $k$ be a positive integer.
The $k$-simulation game $\Gk(\A,\B)$ between Spoiler and Duplicator is defined as follows.

The game is played on arena $Q_\A\times (Q_\B)^{\leq k}$.
The initial position is $(q_\A^0, \{q_\B^0\})$.

A round from position $(p,X)$ consists in the following moves:
\begin{itemize}
\item Spoiler plays a transition $(p,a,p')\in\Delta_\A$
\item Duplicator chooses $X'\subseteq\Delta_\B(X,a)$, with $|X'|\leq k$.
\item the game moves to position $(p',X')$.
\end{itemize}
\noindent
A position $(p,X)$ is winning for Spoiler if $p\in F_\A$ but $X\cap F_\B=\emptyset$. Duplicator wins any play avoiding positions that are winning for Spoiler.

\begin{defiC}[\cite{Ete02}]
The $k$-simulation relation $\A\inck \B$ is said to hold if Duplicator wins $\Gk(\A,\B)$.
\end{defiC}

We visualize positions of the game via \emph{pebbles}: in a position $(p,X)$, Spoiler has a pebble in the state $p$ of $\A$, while Duplicator has pebbles in each state of the set $X$. Notice that according to the definition of the game, Duplicator can duplicate pebbles while erasing some others, as long as it owns at most $k$ pebbles at every step. In other words it is not required that each pebble follows a particular run of the automaton.
The relations $\inck$ can be used to approximate inclusion. Indeed, we have for any NFAs $\A,\B$ \cite{Ete02}:

$$\A\inc_1\B \Rightarrow \A\inc_2\B\Rightarrow \dots \Rightarrow \A\inc_{|\B|}\B \Leftrightarrow L(\A)\subseteq L(\B).$$

Moreover, for fixed $k$ the $\inck$ relation can be computed in polynomial time:

\begin{thm}\label{thm:multiexp} \cite{Ete02}
There is an algorithm with inputs $\A,\B,k$ deciding whether $\A\inck\B$ with time complexity $n^{O(k)}$, where $n=|\A|+|\B|$.
\end{thm}
\noindent
We show in Section \ref{sec:exptime} that this problem is EXPTIME-complete, so this algorithm is optimal in the sense that it cannot avoid the exponent in $k$.

\subsection{Width versus $k$-simulation}
\label{sec:simul}
Links between width and $k$-simulation relations are explicited by the following two lemmas.

The first one shows that knowing the width of an NFA allows to use multipebble simulation to test for real inclusion of languages.

\begin{lem}\label{lem:iNCAidth}
Let $\A,\B$ be NFAs and $k=\w(\B)$. Then $L(\A)\subseteq L(\B)$ if and only if $\A\inck \B$.
\end{lem}

\begin{proof}
The right-to-left implication is true regardless of the value of $k$. It is already stated in \cite{Ete02}, and follows from the fact that if $\A\inck\B$, then any accepting run of $\A$ chosen by Spoiler can be answered with a set of runs from $\B$ containing an accepting one.

We show the converse, and assume $L(\A)\subseteq L(\B)$. Let $\sigma$ be a winning strategy for Player $0$ in $\Gw(\B,k)$, witnessing that $k=\w(\B)$. Then Duplicator can also play $\sigma$ in $\Gk(\A,\B)$, ignoring the position of the pebble of Spoiler in $\A$. Since any word reaching an accepting state of $\A$ is in $L(\A)\subseteq L(\B)$, the strategy $\sigma$ guarantees that at least one pebble of Duplicator is in an accepting state, by definition of $\sigma$. So this strategy is winning in $\Gk(\A,\B)$, witnessing $\A\inck\B$.
\end{proof}

The second lemma shows a link in the other direction: width can be computed from the relations $\inck$.

\begin{lem}\label{lem:wsim}
Let $\A$ be an NFA and $\Adet$ be a DFA for $L(\A)$.
Then for any $k\geq 1$, we have $\w(\A)\leq k$ if and only if $\Adet\inck\A$.
\end{lem}

\begin{proof}
Moves of Spoiler and Duplicator in $\Gk(\Adet,\A)$ are in bjection with those of Player $1$ and Player $0$ respectively in $\Gw(\A,k)$. Indeed, for Spoiler, choosing a transition in $\Adet$ amounts to only choosing a letter, since the state of $\Adet$ is updated deterministically. Moves of Duplicator and Player $0$ are identical in both games. The winning conditions also match in both games, since in a given round, the current state of $\Adet$ is accepting if and only if the word played so far is in $L(\Adet)=L(\A)$.
Therefore, Duplicator wins $\Gk(\Adet,\A)$ if and only if Player $0$ wins $\Gw(\A,k)$, using the same strategy.
\end{proof}

Notice that this does not imply a polynomial reduction between the width problem and multipebble simulation one way or another, since the size of $\Adet$ is in general exponential in the size of $\A$. 

We recall that an automaton accepting all words is said \emph{universal}.

\begin{cor}\label{cor:univ}
Let $\A$ be a universal NFA, and $k\geq 1$. We have $\w(\A)\leq k$ if and only if $\Atriv\inck \A$, where $\Atriv$ is the trivial one-state automaton accepting all words.
\end{cor}

This corollary means that computing the width $k$ of a universal NFA is as hard as testing its multipebble simulations up to $k$ against $\Atriv$.

We will make use of this connection in the following, to show that both the width problem and the multipebble simulation testing are EXPTIME-complete.

\section{Complexity results on the width problem}
\label{sec:exptime}
\newcommand{\Gc}{\GG_{c}}

\newcommand{\init}{\mathit{init}}
\newcommand{\true}{1}
\newcommand{\false}{0}
\newcommand{\bool}{\{\false,\true\}}
\newcommand{\ov}[1]{\overline{#1}}
\newcommand{\Lit}{\mathit{Lit}}
In this section, we study the complexity of the \emph{width problem}: given an NFA $\A$ and an integer $k$, is 
 $\w(\A)\leq k$ ?

Being able to solve this problem efficiently would allow us to optimize the incremental determinisation algorithm, by aiming at the optimal $k$ matching the width right away instead of trying different width candidates incrementally.

The main theorem of this section is the following:
\begin{thm}\label{thm:EXPc} The width problem is EXPTIME-complete.
\end{thm}
\noindent
We start by showing the upper bound:

\begin{lem}
The width problem is in EXPTIME.
\end{lem}
\begin{proof}
To show the EXPTIME upper bound, it suffices to build the game $\Gw(\A,k)$ of exponential size. This is a safety game, so solving it is polynomial in the size of the game. This means this algorithm runs in exponential time. Also note that Algorithm \ref{alg:incrNFA} given in Section \ref{sec:progNFA} computes the width of an NFA in EXPTIME.
\end{proof}

The rest of the section is devoted to showing the EXPTIME-hardness of the width problem.
We will actually show a stronger result: the width problem is EXPTIME-hard on universal safety automata, i.e. automata with all states accepting, and where all words are accepted.
By Corollary \ref{cor:univ}, this implies that this EXPTIME-hardness result applies also to deciding whether a multipebble simulation holds.
We will proceed by reduction from a combinatorial game on boolean formulas, shown EXPTIME-complete in \cite{SC79} (where it is named the game $G_1$). We will call this combinatorial game $\Gc$, and we start by describing it in the following.

\subsection{The combinatorial game $\Gc$}

An instance of the game $\Gc$ is a tuple $(\varphi,X_0,X_1,\alpha_\init)$, where $X_0$ and $X_1$ are disjoint sets of variables, and $\varphi$ is a $4$-CNF formula on variables $V=X_0\cup X_1\cup\{t\}$, where $t\notin X_0\cup X_1$, and finally $\alpha_\init$ is a valuation $V\to\bool$.

This means that $\varphi$ is of the form $C_1\wedge C_2\wedge\dots\wedge C_n$, where $C_i=l_{i,1}\vee l_{i,2}\vee l_{i,3}\vee l_{i,4}$, in which each $l_{i,j}$ is a \emph{literal}, i.e. a variable $x\in V$ or its negation $\ov{x}$.
We will call $\ov{V}$ the set $\{\ov{v}\mid v\in V\}$, and $\Lit=V\cup\ov{V}$ the set of literals. Similarly, we define $\Lit_i=X_i\cup\ov{X_i}$ for $i\in\{0,1\}$. 

A position in $\Gc$ is of the form $(\tau,\alpha)$, where $\tau\in\{0,1\}$ identifies the player who owns the position, and $\alpha$ is a valuation $V\to\bool$.

In such a position, the player $\tau$ owning the position can change the values of variables in $X_\tau$, and additionally the variable $t$ is set to $\tau$. This yields a new valuation $\alpha'$. If this valuation makes the formula $\varphi$ false, Player $\tau$ immediately loses, otherwise the game moves to position $(1-\tau,\alpha')$.

The starting position of the game is $(1,\alpha_\init)$. We say that Player $1$ wins the game if he can force a win, i.e. if he has a strategy $\sigma$ such that all plays compatible with $\sigma$ eventually end with Player $0$ losing the game by making the formula $\varphi$ false.

It is shown in \cite{SC79} than determining whether Player $0$ wins a given instance of the game $\Gc$ is EXPTIME-complete.

\subsection{Reduction to the width problem}

We now want to show that $\Gc$ can be encoded in the width problem for a universal safety automaton. Let $I_c=(\varphi,X_0,X_1,\alpha_\init)$ be an instance of $\Gc$.
We want to build an instance $\A,k$ of the width problem such that $\w(\A)\leq k$ if and only if Player $0$ wins $\Gc$ on instance $I_c$. Moreover the instance $\A,k$ must be computable in polynomial time from $I_c$, and we want $\A$ to be a universal safety automaton.
We will reuse the notations of the previous section for describing the instance $I_c$. In particular $V=X_0\cup X_1\cup \{t\}$, and $\varphi=C_1\wedge C_2\wedge\dots\wedge C_n$, where $C_i=l_{i,1}\vee l_{i,2}\vee l_{i,3}\vee l_{i,4}$. 
The width we will be aiming for is $k=|V|$, i.e. the number of variables in $\varphi$.

\newcommand{\val}{\mathit{val}}

\subsubsection{Intuitive account of the construction}

Before giving a formal definition of $\A$, let us sketch the intuitions guiding this construction.
We want the width game of $\A$ to mimic the game $\Gc$. This means that a subset of states of $\A$ of size $k$ will correspond to a valuation of $k$ variables.
Truth value of variables from $X_0$ will be chosen by Player $0$ through nondeterminism in $\A$, while truth value for variables from $X_1$ will be chosen by Player $1$ via his choice of letters in the width game.
Gadgets will then be added to 
\begin{itemize}
\item control whether the current valuation makes the formula $\varphi$ true,
\item set the initial valuation to $\alpha_\init$ and
\item make the automaton universal.
\end{itemize}

\noindent The width game should be lost immediately by the player who chose a valuation making the formula false.

These components will first be built as an auxiliary automaton $\B$. 
In order to properly mimic the dynamic of the game $\Gc$, one needs to additionally constrain the letters that can be played by Player $1$. This will be done by a separate deterministic automaton $\C$, forcing the word played by Player $1$ to belong to a certain language.
Finally, $\A$ will be obtained by a cartesian product of $\B$ and $\C$. 
In order to illustrate the construction of $\B$ and $\C$, we will use a running example:

\begin{exa}\label{ex:run}
The instance $I_c$ of $\Gc$ we take as example is $(\varphi,X_0,X_1,\alpha_\init)$, with $\varphi=(x\vee y\vee z\vee t)\wedge (\ov{x}\vee y \vee \ov{z}\vee \ov{t})$, $X_0=\{x,y\}$, $X_1=\{z\}$, and $\alpha_{\init}$ to be the valuation setting all variables to true. Notice that this instance is won by Player $0$, as it suffices to always set $y$ to true to guarantee that $\varphi$ remains true.
\end{exa}

After defining $\B$ and $\C$ in Sections \ref{sec:B} and \ref{sec:C}, we will prove in Section \ref{sec:proofcor} that the game $\Gw(\A,k)$ correctly emulates the game $\Gc$ on instance $I_c$.

\subsubsection{The automaton $\B$}
\label{sec:B}

The automaton $\B$ is the main gadget of the construction. The idea is that moving $k$ pebbles in $\B$ will be equivalent to choosing a valuation of $k$ variables, via a set of $2k$ states (called $Q_\Lit$): one state for every literal, corresponding to a truth value assignation for its associated variable.
The transitions of $\B$ are designed so that for variables in $X_0$, Player 0 can choose the valuation using the nondeterminism of $\B$ on a single letter $a$, while for variables in $X_1$, Player $1$ chooses a valuation by choosing which letters of the form $f_l$ to play.

Let us start by describing the automaton $\B$ associated to the running example \ref{ex:run}.
\begin{exa}

\begin{figure}[h]
\centering

\scalebox{.9}{
\begin{tikzpicture}[->,>=stealth',shorten >=1pt,auto,node distance=2.8cm,
                    semithick]
  \tikzstyle{every state}=[fill=none,draw=black,text=black]

  \node[state,initial] (q_0)                    {$q_0$};
  \node[state]         (x) [above left=2cm and 2cm of q_0] 	  {$q_x$};
  \node[state]         (nx) [right=1cm of x] 	  {$q_{\ov{x}}$};
  \node[state]         (y) [right=1cm of nx] 	  {$q_y$};
  \node[state]         (ny) [right=1cm of y] 	  {$q_{\ov{y}}$};
  \node[state]         (z) [below left=2cm and 1cm of q_0] 	  {$q_z$};
  \node[state]         (nz) [right=1cm of z] 	  {$q_{\ov{z}}$};
  \node[state]         (t) [above right=.1cm and 3cm of q_0] 	  {$q_t$};
  \node[state]         (nt) [below=1cm of t] 	  {$q_{\ov{t}}$};
	

 \path 
 (q_0) edge[bend left] node [left] {$a$} (x)
 	edge node [left] {$a$} (y)
 	edge node [right] {$a$} (z)
 	edge node [above] {$a$} (t)
 	edge node [right] {$a$} (nx)
 	edge node [below] {$a$} (nt)
 	edge[bend right, out=0, in=-160] node [right] {$a$} (ny)
 	edge node [left] {$a$} (nz)
 	
 	
 (x) edge[bend left,<->] node{$a$} (nx)
 	edge[loop above] node{$a$} () 
 (nx) 
 edge[bend left] node {$d_x$} (x) 
 edge[loop above] node{$a$} ()
 	
  (y) edge[bend left,<->] node {$a$} (ny)
 	edge[loop above] node{$a$} ()
 (ny) edge[loop above] node{$a$}()
 edge[bend left] node[pos=.55] {$d_y$} (y) 
 	
 	 (z) edge[bend left] node{$f_{\ov{z}}$} (nz)
 (nz) edge[bend left] node{$d_z,f_{z}$} (z)
 	
 	 (t) edge[bend left] node{$a$} (nt)
 (nt) edge[bend left] node{$d_t,f_t$} (t)
;

%
%


  \node[state]         (X) [below left=2cm and 3cm of z] 	  {$q_x$};
  \node[state]         (nX) [right=1cm of X] 	  {$q_{\ov{x}}$};
  \node[state]         (Y) [right=1cm of nX] 	  {$q_y$};
  \node[state]         (nY) [right=1cm of Y] 	  {$q_{\ov{y}}$};
  \node[state]         (Z) [right=1cm of nY] 	  {$q_z$};
  \node[state]         (nZ) [right=1cm of Z] 	  {$q_{\ov{z}}$};
  \node[state]         (T) [right=1cm of nZ] 	  {$q_t$};
  \node[state]         (nT) [right=1cm of T] 	  {$q_{\ov{t}}$};

\path
(X) edge node {$\begin{array}{r}c_1,d_x\\ e_x,a_{\ov{x}}\end{array}$} ++(0,1)
 edge[dashed] node {$c_2$} ++(0,-1)
 (nX) edge node {$\begin{array}{r}c_2,d_{\ov{x}},\\e_x,a_x\end{array}$} ++(0,1)
 edge[dashed] node {$c_1$} ++(0,-1)

 (Y) edge node {$\begin{array}{r}c_1,c_2,d_y\\e_y,a_{\ov{y}}\end{array}$} ++(0,1)
 (nY) edge node {$\begin{array}{r}d_{\ov{y}}\\e_y,a_y\end{array}$} ++(0,1)
 edge[dashed] node {$c_1,c_2$} ++(0,-1)
 
 (Z) edge node {$\begin{array}{r}c_1,d_z\\e_z,a_{\ov{z}}\end{array}$} ++(0,1)
 edge[dashed] node {$c_2$} ++(0,-1)
 (nZ) edge node {$\begin{array}{r}c_2,d_{\ov{z}}\\e_z,a_z\end{array}$} ++(0,1)
 edge[dashed] node {$c_1$} ++(0,-1)
 
 (T) edge node {$\begin{array}{r}c_1,d_t\\e_t,a_{\ov{t}}\end{array}$} ++(0,1)
 edge[dashed] node {$c_2$} ++(0,-1)
 (nT) edge node {$\begin{array}{r}c_2,d_{\ov{t}}\\e_t,a_t\end{array}$} ++(0,1)
 edge[dashed] node {$c_1$} ++(0,-1)
 
;
\end{tikzpicture}
}
\caption{Valuation gadget and exit transitions of the automaton $\B$\label{fig:B}}
\end{figure}
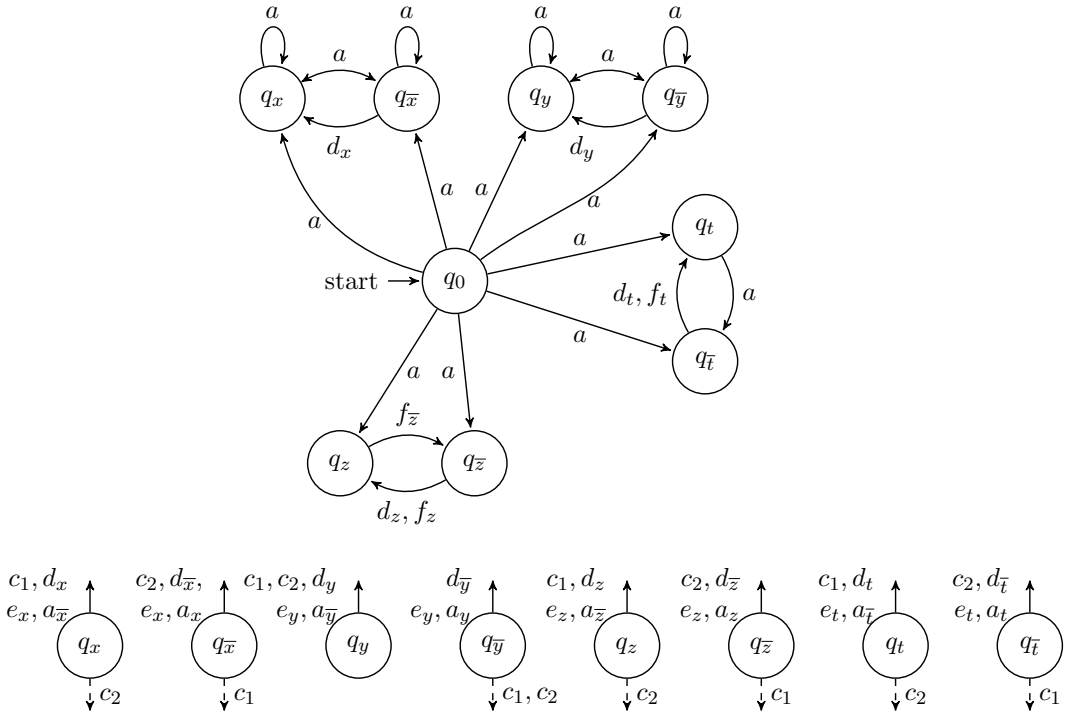

The automaton $\B$ corresponding to Example \ref{ex:run} is described in Figure \ref{fig:B}. The first diagram represents initial transitions and transitions changing the valuation. Deterministic self-loops such as $q_x\trans{d_{\overline{x}}} q_x$ are omitted for readability. The second diagram describes for each state $q_l$ which letters from $\{c_1,c_2\}$ cannot be read (dashed arrows to the bottom), and which letters go to the accepting sink $q_\top$ (arrows to the top). The automaton $\B$ is safe, i.e. all states are accepting.
The only way for a run to fail is to read a letter $c_i$ in a state where it is forbidden. These letters are used to control that the valuation chosen by Player $0$ makes the formula true. For instance reading $c_2$ in $q_x$ leads to a fail of the run, because $x$ does not make clause $2$ true in the formula $\varphi$. If Player $1$ can play a letter $c_i$ such that no pebble is in a state $q_l$ where $l$ makes clause $c_i$ true,   he immediately wins the game.
\end{exa}

In this example, we can notice that the letter $a$ is used to allow Player $0$ to set the values of variables $x$ and $y$ via nondeterminism. On the other hand, the letters $f_z$ and $f_{\overline{z}}$ can be played by Player $1$ to set the value of variable $z$. Letters $d_l$ for each literal $l$ have two roles: they set the initial valuation  (here with all variables to true), and they will help to guarantee that the final automaton is universal, by immediately leading to an accepting state if letter $d_l$ is read in state $q_l$.

We now give the formal definition of the general construction of $\B=(Q,\Sigma,q_0,\Delta,F)$.
\medskip

\noindent \textbf{States}

Let $Q_\Lit=\{q_l\mid l\in \Lit\}$. The states in $Q_\Lit$ will be used to encode valuations $\alpha$, via the positions of the pebbles in the game $\Gw(\A,k)$.

$Q_\Lit$ is partitioned into $Q_0,Q_1,Q_t$, with $Q_i=\{q_l \mid l\in \Lit_i\}$ for $i\in\{0,1\}$ and $Q_t=\{t,\ov{t}\}$.


We finally set $Q=\{q_0,q_\top\}\cup Q_\Lit$ where $q_0$ is the initial state and $q_\top$ is an accepting sink state. Reaching $q_\top$ with one of its pebble will mean immediate win for Player $0$, as this pebble will trivially accept any word subsequently played by Player $1$. We define $F=Q$, i.e. $\B$ is a safety automaton, and every run is accepting.
Notice that the number of states of $\B$ is $|Q|=2+2k$, so it is polynomial in the size of the instance $I_c$ of $\Gc$.

\medskip
\noindent \textbf{Alphabet}

We define here several sub-alphabets that will be used in our encoding.
For each of them, we already give an intuition of how it will be used.
The alphabet is presented in two groups, separated by a double line: the first group will be used in the normal flow of the simulation of the game, while the second group is used by Player $1$ to challenge choices of Player $0$, and will normally never be played if Player $0$ is playing a correct strategy.
\medskip

\noindent\begin{tabular}{|p{.3\textwidth}p{.67\textwidth}|}
\hline
$\{a,f_t\}$ & The letter $a$ allows Player $0$ to choose a value for all variables from $X_0$, and sets variable $t$ to false. The letter $f_t$ is used to set variable $t$ to true.\\
\hline
$\Gamma_\Lit:=\{a_l\mid l\in\Lit\}$ &
For each clause, Player $1$ will have to play a letter $a_l$ witnessing that his valuation makes this clause true.\\
\hline

$\Gamma_1:=\{f_l\mid l\in \Lit_1\}$ &
The letter $f_l$ will be played by Player $1$ to set the literal $l$ to true.\\
\hline
$\Sigma_D:=\{d_l\mid l\in\Lit\}$ & Used just once at the beginning: sets initial valuation while making the automaton universal.\\
\hline
\hline
$\Sigma_C:=\{c_i\mid i\in[1,n] \}$ &
A letter $c_i$ will be played by Player $1$ if the valuation chosen by Player $0$ fails to make the clause $C_i$ true.\\
\hline
$\Sigma_V:=\{e_v\mid v\in V\}$ &
The letter $e_v$ will be played by Player $1$ if Player $0$ has failed to set a value for variable $v$.\\
\hline
\end{tabular}
\medskip

We set $\Sigma=\{a,f_t\}\cup\Gamma_\Lit \cup \Gamma_1\cup\Sigma_D\cup\Sigma_C\cup\Sigma_V$.

Notice that $|\Sigma|\leq 2+2k+2k+2k+n+k=n+7k+2$, so it is polynomial in the size of $I_c$ as well.

\medskip
\noindent \textbf{Transitions}


We will use the notation $p\trans{a} q$ to mean that we put a transition $(p,a,q)$ in $\Delta$.
If $l\in\Lit$ is a literal, we define its projection $\pi(l)$ to variables by $\pi(l)=v$ if $l\in\{v,\ov{v}\}$.
We define the negation of a literal $l\in\{v,\ov{v}\}$ as $\ov{l}=\ov{v}$ if $l=v$ and $\ov{l}=v$ if $l=\ov{v}$.
\medskip

We present the transition table $\Delta$ in the following array: the left column contains the set of transitions, while the right column explains their roles in the reduction. 
As before, we separate in a second component transitions that will only be taken in the event of a Player failing to play a valid strategy, thereby terminating the simulation of the game.
\medskip

\noindent\begin{tabular}{|p{.3\textwidth}p{.66\textwidth}|}
\hline
$q_0\trans{a} q_l$ for all $l\in\Lit$&
allows Player $0$ to choose any initial valuation.\\
\hline
\parbox{.3\textwidth}{$q_l\trans{d_{l'}} q_{l''}$\\ if $l\neq l'$ and $l''=\alpha_\init(\pi(l))$}&
sets the initial valuation to $\alpha_\init$.\\
\hline
\parbox{.3\textwidth}{$q_l\trans{f_{l'}} q_{l'}$ if $\pi(l)=\pi(l')$,\\ for all $l'\in\Lit_1\cup\{t\}$}&
sets the value of $\pi(l)$ to $l'$.\\
\hline
\parbox{.3\textwidth}{$q_l\trans{f_{l'}} q_{l}$ if $\pi(l)\neq\pi(l')$,\\ for all $l'\in\Lit_1\cup\{t\}$}&
leaves other variable unchanged.\\
\hline
\parbox{.3\textwidth}{$q_l\trans{a} q_{l'}$ \\if $l\in \Lit_0$ and $\pi(l)=\pi(l')$,\\ or if $l\in\Lit_1$ and $l=l'$}&
\parbox{.66\textwidth}{nondeterministic choice for variables in $X_0$,\\ leaving variables from $X_1$ unchanged.}\\
\hline
 $q_l\trans{a} q_{\ov{t}}$ if $\pi(l)=t$&
sets $t$ to false during the turn of Player $0$.\\
\hline
 $q_l\trans{a_l} q_l$ if $l\in\Lit$&
will be used to validate the valuation chosen by Player $1$.\\
 \hline
 \hline
  $q_l\trans{d_{l}} q_\top$ for all $l\in\Lit$&
helps towards universality of the automaton: a good guess from Player $0$ on which $d_l$ will be played leads to immediate acceptation.\\
\hline
$q_l\trans{e_{\pi(l)}}q_\top$ for all $l\in\Lit$&
$e_v$ is played if no value is set for variable $v$, leading to instant loss for Player $0$ if no value is set and instant win otherwise.\\
\hline
\parbox{.3\textwidth}{$q_l\trans{c_i} q_\top$\\ if literal $l$ appears in $C_i$}&
if Player $1$ challenges the valuation with clause $C_i$, instant win for Player $0$ if a literal makes the clause true.\\
\hline
 $q_l\trans{a_{\ov{l}}} q_\top$ if $l\in\Lit$&
if Player $1$ tries to validate his valuation with a wrong literal, instant win for Player $0$.\\
\hline
 $q_\top\trans{b} q_\top$ for all $b\in\Sigma$ &
 accepting sink state.\\
 \hline
\end{tabular}
\medskip

This achieves the definition of $\B$. Notice that $\B$ has size polynomial in the size of $I_c$, and can be computed from $I_c$ in polynomial time.

\subsubsection{The automaton $\CC$}
\label{sec:C}

The automaton $\CC$ is used to restrict the moves of Player $1$, i.e. the letters chosen in the width game, to those that are relevant to the game $\Gc$. This includes for instance forcing him to prove that his own valuations make the formula true via the letters $a_l$, and allowing him to challenge valuations chosen by Player $0$ via the letters $c_i$. We define a safety language such that if Player $1$ plays a bad prefix of this language, then the whole automaton $\A$ immediatly goes to an accepting sink state, and therefore Player $0$ wins the width game.

We formally describe $\CC$ in the following, and instantiate it on the running example. 

For each $i\in[1,n]$, let $A_i=\{a_l\mid \text{literal $l$ appears in }C_i\}$.

Let $L_\val=A_1A_2\dots A_n$, $L_\val$ is a subset of $(\Gamma_\Lit)^{n}$.
The purpose of $L_\val$ is to check that the valuation chosen by Player $1$ (corresponding to Player $1$ in $\Gc$) is valid, by forcing him to choose one valid literal by clause. 
\medskip

We define  $$L_\CC=a\Sigma_D(\varepsilon+\Sigma_V)\big(\Gamma_1^{|X_1|} \cdot f_t \cdot L_\val\cdot a(\varepsilon+\Sigma_V)(\varepsilon+\Sigma_C)\big)^*.$$

We detail in the following the meaning of the different factors in this expression, by order of appearance:
\medskip

\noindent\begin{tabular}{|p{.3\textwidth}p{.66\textwidth}|}
\hline
$a\Sigma_D$ &
sets the initial configuration, and makes the automaton accept anything\\
\hline
$(\varepsilon+\Sigma_V)$ &
punishes Player $0$ if a variable is not instantiated\\
\hline
$\Gamma_1^{|X_1|}\cdot f_t$ &
chooses values for variables in $X_1$ and sets $t$ to true\\
\hline
$L_\val$ &
proves that the chosen valuation makes $\varphi$ true\\
\hline
$a$&
allows Player $0$ to choose a valuation for $X_0$, and sets $t$ to false\\
\hline
$(\varepsilon+\Sigma_V)$ &
verifies that each variable is still instantiated ($x$ and $\overline{x}$ could be set to true while erasing the valuation of $y$).\\
\hline
$(\varepsilon+\Sigma_C)$ &
allows to challenge the valuation with a clause $C_i$. \\
\hline
\end{tabular}
\medskip

Let $\CC_0$ be a complete DFA recognizing $L_\CC$ with rejecting sink state $\top_\CC$. We build $\CC$ by making all states of $\CC_0$ accepting, including its sink state $\top_\CC$. This makes $\CC$ a universal complete DFA, and any word that is not a prefix of a word in $L_\CC$ will reach the accepting sink state $\top_\CC$. It is straightforward to build $\CC$ in polynomial time from $I_c$, with the size of $\CC$ polynomial in the size of $I_c$. 

\begin{exa}
The deterministic safety automaton $\CC$ corresponding to the running example \ref{ex:run} is represented in Figure \ref{fig:C}. Some states are labeled by the type of the position they represent, see next section. Transitions to the accepting sink $\top_\CC$ are not represented: all missing transitions for this deterministic automaton to be complete go to state $\top_\CC$.
\begin{figure}[h]
\centering

\scalebox{.9}{
\begin{tikzpicture}[->,>=stealth',shorten >=1pt,auto,node distance=1cm,
                    semithick]
  \tikzstyle{every state}=[fill=none,draw=black,text=black]

  \node[state,initial] (A1) {};
  \node[state,right=of A1] (A1') {};
  \node[state, right= of A1'] (A2) {1};
  \node[state, right= of A2] (A3) {};
  \node[state, below =of A2] (A4) {};
  
  \node[state, right= of A3] (A8) {};
  \node[state, right=2cm of A8] (A9) {};
  \node[state, below=of A9] (A10) {0};
  \node[state, below left= of A10] (A11) {1};
  \node[state, left=1.7cm of A11] (A13) {};

  \node[state, below= of A1'] (top) {$\top_\CC$};
\path
(A1) edge node{$a$} (A1')
(A1') edge node{$\Sigma_D$} (A2)
(A2) edge node[left]{$\Sigma_V$} (A4)
(A2) edge node{$f_z,f_{\ov{z}}$} (A3)
(A4) edge node[near start,right]{$f_z,f_{\ov{z}}$} (A3)
(A11) edge node[ right]{$f_z,f_{\ov{z}}$} (A3)
(A13) edge node[right, near start] {$f_z,f_{\ov{z}}$} (A3)

(A3) edge node{$f_t$} (A8)
(A8) edge node{$a_x,a_y,a_z,a_t$} (A9)
(A9) edge node{$a_{\ov{x}},a_y,a_{\ov{z}},a_{\ov{t}}$} (A10)
(A10) edge node{$a$} (A11)
(A11) edge node[above]{$\Sigma_V$} (A13)
(A13) edge node[above]{$c_1,c_2$} (A4)
(A11) edge[bend left, out=60, in=130] node{$c_1,c_2$} (A4)
;
\end{tikzpicture}
}
\caption{The complete safety DFA $\CC$\label{fig:C}}
\end{figure}
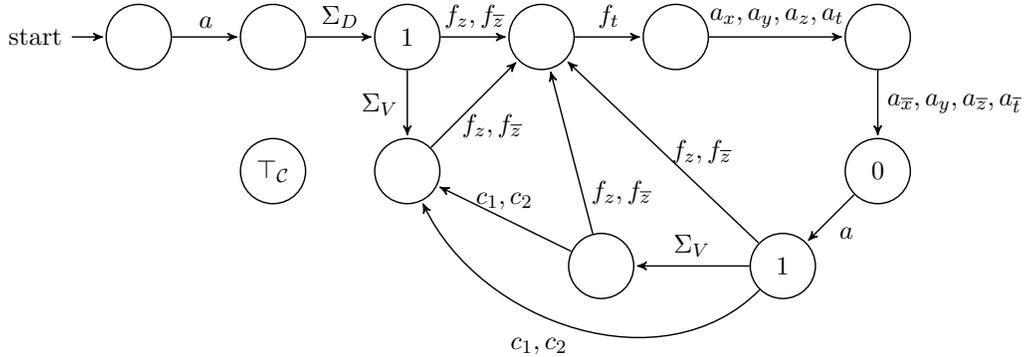
\end{exa}
\medskip

\subsubsection{The main automaton $\A$}

We now combine $\B$ and $\CC$ to obtain the automaton $\A$, for which being able to compute width will amount to solve the game $\Gc$.

The automaton $\A$ is defined as the cartesian product of $\B$ and $\CC$, with the additional modification that all states of the form $(q,\top_\CC)$ or $(q_\top,q')$ are merged to a unique accepting sink state $\top_\A$. We also add all transitions of the form $(q,q')\trans{x}\top_\A$ as soon as $q'\trans{x}\top_\CC$ is a transition of $\CC$, even when the letter $x$ cannot be read from $q$ in $\B$. 

Since $\CC$ is deterministic, it has no impact on the width, and the only non-determinism for Player $0$ to resolve in $\Gw(\A,k)$ comes from $\B$.

\begin{lem}
$\A$ is a universal safety automaton.
\end{lem}

\begin{proof}
Notice that since $\B$ and $\CC$ are safety automata, $\A$ is a safety automaton as well.
We have to show that $\A$ accepts all words.

Let $w\in\Sigma^*$. If $w$ truncated to its first two letters is not a prefix of $ad_l$ for some $l\in\Lit$, then it immediately ends up in state $\top_\CC$ in $\CC$, and so it reaches $\top_\A$ in $\A$. So in this case, $w\in L(\A)$. Assume now that $w=ad_lv$ for some $l\in\Lit$ and $v\in\Sigma^*$.

Then the run $q_0\trans{a}q_l\trans{d_l} q_\top \trans{v} q_\top$ of $\B$ is a witness that $\A$ can reach $\top_\A$ when reading $w$. 

Finally, the words $\varepsilon$ and $a$ are also accepted by $\A$.

Therefore, any $w\in\Sigma^*$ is accepted by $\A$, and $\A$ is a universal safety automaton.
\end{proof}




We are now ready to prove that the construction of $\A$ performs the wanted reduction, thereby completing the proof of Theorem \ref{thm:EXPc}.

\subsection{Proof of correctness}
\label{sec:proofcor}

We prove in this section that the above construction allows to use the width game of $\A$ to emulate the game $\Gc$.

Let $k=|V|$, we want to prove that $\w(\A)>k$ if and only if Player $1$ wins in $I_c$.

We will show that the game $\Gw(\A,k)$ simulates the game $\Gc$, by establishing a correspondence between strategies for these two games.

\medskip
\noindent \underline{Player $1$ wins in $I_c$ $\Longrightarrow$ $\w(\A)>k$.}

Assume Player $1$ wins in $I_c$, with a winning strategy $\sigma_c$. We aim at building a winning strategy $\sigma_w$ for Player $1$ in $\Gw(\A,k)$. It means that Player $1$ can enforce a position of the game where the word played is in $L(\A)=\Sigma^*$, but all pebbles have been erased due to non-existing transitions.

We now define $\sigma_w$. Following the language $L_C$, Player $1$ starts by playing $a$.
Player $0$ can move the pebbles to any subset of $Q_\Lit$ of size at most $k$. In order to prevent Player $0$ from reaching $\top_\A$, Player $1$ will now play $d_l$ where $q_l$ is a state not occupied by a pebble. This puts all pebbles to the states corresponding to the initial valuation.

If not all variables are instantiated by a pebble, Player $1$ plays $e_v\in\Sigma_v$ such that $q_v$ and $q_{\ov{v}}$ do not contain a pebble. The resulting word $ad_le_v$ is in $L(\A)=\Sigma^*$, but Player $0$ fails to accept it, since no state occupied by a pebble can read $e_v$, so Player $1$ wins the game. We can therefore assume that Player $0$ uses all $k$ pebbles and reaches all states $q_l$ corresponding to the valuation $\alpha_\init$. In this case we define strategy $\sigma_w$ so that Player $1$ does not play a letter in $\Sigma_V$, as allowed by the $\varepsilon$ in the definition of $L_\CC$.

We now switch to the main dynamic of the game, and we will match certain positions of $\Gw(\A,k)$ to positions of the game $I_c$. The current position of $\Gw(\A,k)$, where $a\cdot d_l$ has been played and the $k$ pebbles are in the states of $Q_\Lit$ matching $\alpha_\init$, corresponds to the initial position $(1,\alpha_\init)$ of $I_c$.

More generally the positions reached after playing a word from $L_\CC$ will be matched to positions $(1,\alpha)$ of $I_c$, where $\alpha$ is described by the states $q_l$ occupied by the $k$ pebbles. We say that such a position of $\Gw(\A,k)$ is of type $1$. Similarly the positions reached after playing a word from $L_\CC(\Gamma_1^{|X_1|} \cdot f_t \cdot L_\val)$ will be matched to a position $(0,\alpha')$ of $I_c$, and are called positions of type $0$. 

We define $\sigma_w$ in positions of type $1$ in the following way: Let $\alpha_1$ be the values of variables in $X_1$ chosen by $\sigma_c$ in the matching position of $I_c$. The factor of $\Gamma_1^{|X_1|}$ played by $\sigma_w$ will explicit these values, by playing for each literal $l$ that is true in $\alpha_1$ the letter $f_l$. This switches the pebbles in $X_1$ to match the valuation $\alpha_1$, and leave the other variables (from $X_0\cup\{t\}$) unchanged.

The letter $f_t$ must then be played by Player $1$, in order to avoid losing by allowing Player $0$ to put his pebbles in $\top_\A$ (any other letter would lead the deterministic run of $\CC$ to $\top_\CC$). This moves the $t$ pebble to $q_t$, setting the value of $t$ to $1$, according to the definition of $\Gc$.

Player $1$ must now play a word in $L_\val$. Since $\sigma_c$ is winning, the current valuation $\alpha$ makes the formula true. The strategy $\sigma_w$ consists in witnessing this by choosing for each clause $C_i$ a literal $l$ from $\alpha$ that makes it true, and play $a_l$. This leaves the position of the pebbles unchanged.

We have now reached a position of type $0$. We define the strategy $\sigma_w$ for these positions.
First, Player $1$ must play the letter $a$. This allows Player $0$ to move pebbles from $Q_0$ freely, thereby choosing a new valuation for variables in $X_0$. Moreover it moves a pebble from $q_t$ to $q_{\ov{t}}$, setting the value of $t$ to $0$. Notice that by the definition of $\Gw(\A,k)$, Player $0$ could also duplicate some pebbles and erase some others, thereby setting some variables in $X_0$ to both true and false, and not assigning other variables from $V$.
If Player $0$ chooses to do this, the strategy $\sigma_w$ of Player $1$ will immediately punish it by playing $e_v$ where $v$ is a non-assigned variable, and as before this allows Player $1$ to win the game.
This allows us to continue assuming the pebbles describe a valuation $\alpha$ of all variables in $V$.
If $\alpha$ makes the formula true, we are back to a position of type $1$, and we continue with the strategy as described.

On the other hand, if we have reached a winning position for Player $1$ in $I_c$, i.e. if the valuation $\alpha$ makes the formula false, we show that Player $1$ can win $\Gw(\A,k)$. To do so, he plays a letter $c_i$ such that no literal in $C_i$ is true in $\alpha$. This way, no pebble is in a state where $c_i$ can be read, and no pebbles are present in the next position of $\Gw(\A,k)$. Since the word $w$ played until now is in the language $L(\A)=\Sigma^*$, this is a winning position for Player $1$ in $\Gw(\A,k)$.

Since $\sigma_c$ is winning, the game will eventually reach a position where the valuation chosen by Player $0$ makes the formula $\varphi$ false, hence $\sigma_w$ is a winning strategy for Player $1$ in $\Gw(\A,k)$, witnessing $\w(\A)>k$.

\medskip

\noindent \underline{Player $0$ wins in $I_c$ $\Longrightarrow$ $\w(\A)\leq k$.}

%
Assume that Player $0$ has a winning strategy $\sigma_0$ in $I_c$.
This means that this strategy avoids losing positions for Player $0$, either by playing forever, or by reaching a position that is losing for Player $1$. Moreover, since $\Gc$ is a safety game for Player $0$, $\sigma_0$ can be chosen positional, i.e. its move only depends on the current valuation $\alpha$.

We will show that this strategy $\sigma_0$ can be turned into a strategy $\tau_0$ that is winning for Player $0$ in $\Gw(\A,k)$, thereby witnessing $\w(\A)\leq k$.

The strategy $\tau_0$ consists in the following:
\begin{itemize}
\item on the first occurence of $a$, put the pebbles in all states $q_l$ with $l$ appearing in $\alpha_\init$ (or to any other valuation),
\item on other occurences of $a$, match the choice of $\sigma_0$ for the valuation of $X_0$ (according to the current valuation $\alpha$).
\end{itemize}
Other choices to be made by $\tau_0$ are described in the following.

First of all, we may assume that Player $1$ always plays words from $L_\CC$, otherwise he immediately loses $\Gw(\A,k)$, as the automaton $\A$ goes to $\top_\A$.

After the prefix from $a\Sigma_D$, Player $1$ has no interest in playing a letter in $\Sigma_V$, since the strategy $\tau_0$ assigned a value to each variable.

We are then in a position of type $1$, and Player $1$ must play a word in $\Gamma_1^{|X_1|}$. This allows him to choose a valuation for any variable in $X_1$. The letter $f_t$ then sets the value of $t$ to $1$, according to the rules of $\Gc$.

Player $1$ must now play a word from $L_\val$. We show that this forces him to prove that the current valuation $\alpha$ makes the formula $\varphi$ true. Indeed, for each clause $i$, Player $1$ must choose a literal $l$ appearing in $C_i$, and play $a_l$. If $l$ is currently false, i.e. there is a pebble in $q_{\ov{l}}$, Player $0$ can move this pebble to $q_\top$ and wins the game $\Gw(\A,k)$. Therefore, if Player $1$ cannot choose a valuation of $X_1$ setting $\varphi$ to true, Player $0$ wins the game $\Gw(\A,k)$.

Otherwise, we reach a position of type $0$. Letter $a$ allows Player $0$ to choose a valuation for $X_0$, and sets $t$ to $0$. Strategy $\tau_0$ is defined to do this accordingly to $\sigma_0$, and therefore this will always reach a valuation $\alpha$ setting $\varphi$ to true. If Player $1$ plays a letter from $\Sigma_V$, Player $0$ can reach $q_\top$ and win the game. If Player $1$ plays a letter $c_i$ from $\Sigma_C$, it will allow Player $0$ to reach $q_\top$, since one literal $l$ of clause $C_i$ is currently set to true, via a pebble in $q_l$.
Therefore, the only interesting move for Player $1$ is to go back to his move in a position of type $1$, and play a new valuation of $X_1$ via a word in $\Gamma_1^{|X_1|}$.

Since $\sigma_0$ is winning in $I_c$, either the play goes on forever in $I_c$, which means Player $0$ wins the corresponding play in $\Gw(\A,k)$, or Player $0$ is eventually able to reach $q_\top$, either because Player $1$ loses in $I_c$ via a bad valuation, or because he made bad choices in $\Gw(\A,k)$, for instance by playing a letter from $\Sigma_V$. Either way, the strategy $\sigma_0$ is winning for Player $0$ in $\Gw(\A,k)$.
\medskip

This achieves the proof that $\w(\A)>k$ if and only if Player $1$ wins in $I_c$.
Since the reduction from an instance of $I_c$ to an instance of the width problem for a universal safety NFA can be done in polynomial time, we showed the following theorem:

\begin{thm}
Computing the width of a universal safety NFA is EXPTIME-complete.
\end{thm}

\noindent By Corollary \ref{cor:univ}, we obtain the following result:

\begin{cor}
It is EXPTIME-complete to decide, given two NFAs $\A,\B$ and $k\geq 1$, whether $\A\inck\B$. 
This is already true when $\A$ is fixed to the trivial automaton $\Atriv$ and the input $\B$ is restricted to universal safety NFAs.
\end{cor}

\begin{rem}\label{rem:notdup}
Notice that this proof also shows that the alternative versions of width and multipebble simulations where pebbles cannot be duplicated are also EXPTIME-complete. This is because our reduction actually only needs this kind of moves, and uses the $\Sigma_V$ gadget to forbid duplicating pebbles while erasing others. Moreover, all results from Section \ref{sec:simul} can be carried to this alternative version.
\end{rem}

\begin{rem}\label{rem:omega}
Although the present section deals with finite words, most results are immediately transferable to safety automata on infinite words. Any infinite run is accepting in a safety automaton. This acceptance condition is of particular interest in verification, as it describes very natural properties such as deadlock freeness of a system. See \cite{Baier} for an introduction to automata on infinite words for verification.
This also shows that the EXPTIME-hardness result can be lifted to any acceptance condition generalising the safety one, such as B\"uchi, coB\"uchi, parity.
\end{rem}
%
%
%
%
%


\section{CoB\"uchi Automata}\label{sec:cobuchi}

We now turn to the case of coB\"uchi automata, and their determinisation problem. Here, since GFG and DBP are no longer equivalent \cite{BKKS13,KS15}, aiming for a GFG automaton becomes a problem that is different from determinization via DBP automata. We will compare these two problems, and we will see that the class of GFG coB\"uchi automata is particularly interesting for several reasons.

First of all recall that NCA and DCA have same expressive power, i.e. the determinisation of coB\"uchi automata does not need to introduce more complex acceptance conditions. This follows from the breakpoint construction \cite{MH84} that we will generalise in this section to its incremental variant.

\subsection{Width of $\omega$-automata}
We define here the width of automata on infinite words in a general way, as the definition is independent of the accepting condition.

Let $\A=(Q,\Sigma,q_0,\Delta,\alpha)$ be an automaton on infinite words with acceptance condition $\alpha$, and $n=|Q|$ be the size of $\A$.

As before, we want to define the \emph{width} of a $\A$ as the minimum number of states that need to be tracked in order to deterministically build an accepting run in an online way.

We will use the same family of games $\Gw(\A,k)$ as in Section \ref{sec:width}, and they will only differ in the winning condition.

The game $\Gw(\A,k)$ is played on $\Qk$, starts in $X_0=\{q_0\}$, and the round $i$ of the game from a position $X_i\in \Qk$ is defined as follows:
\begin{itemize}
\item Player 1 chooses a letter $a_{i+1}\in \Sigma$.
\item Player 0 moves to a subset $X_{i+1}\subseteq \Delta(X_i,a_{i+1})$ of size at most $k$.
\end{itemize}
\noindent
An infinite play is winning for Player $0$ if whenever $a_1a_2\dots \in L(\A)$, the sequence  $X_0X_1X_2\dots$ contains an accepting run. That is to say there is a valid accepting run $q_0q_1q_2\dots$ of $\A$ on $a_1a_2\dots$ such that for all $i\in\N$, $q_i\in X_i$.

\begin{defi}
The width of $\A$, denoted $\w(\A)$, is the least $k$ such that Player $0$ wins $\Gw(\A,k)$.
\end{defi}

As before, an automaton $\A$ is GFG if and only if $\w(\A)=1$.

\subsection{GFG coB\"uchi automata}
We recall here previous results on GFG coB\"uchi automata.

The first result is the exponential succinctness of GFG NCAs compared to DCAs.
\begin{thmC}[Theorem 1, \cite{KS15}]\label{thm:succinct}
There is a family of languages $(L_n)_{n\in\N}$ such that for all $n$, $L_n$ is accepted by a coB\"uchi GFG automaton of size $n$, but any deterministic parity automaton for $L_n$ must have size in $\Omega\big(\frac{2^n}{n}\big)$. 
\end{thmC}

\noindent Moreover, each language $L_n$ can be chosen LTL-definable \cite{IK19}, hinting towards applicability of GFG NCAs in LTL synthesis.

This means that some GFG NCA only admit GFG strategies with exponential memories \cite{KS15}, i.e. the witness for an NCA being GFG can be of exponential size. Despite this fact, the next theorem shows that GFG NCAs can be recognised efficiently.

\begin{thmC}[Theorem 11, \cite{KS15}]\label{thm:GFG_P}
Given an NCA $\A$, it is in PTIME to decide whether $\A$ is GFG.
\end{thmC}

\noindent It was also shown that GFG coB\"uchi automata can be efficiently minimized:
\begin{thmC}[Theorem 21, \cite{AK19}]\label{thm:minCob}
Cob\"uchi automata with acceptance condition on transitions can be minimized in polynomial time.
\end{thmC}

The conjunction of these results make the coB\"uchi class particularly interesting in our setting: the succinctness allows us to potentially save a lot of space compared to classical determinisation, and Theorem \ref{thm:GFG_P} can be used to stop the incremental construction. Moreover, once a coB\"uchi GFG automata has been built, it can be minimized efficiently thanks to Theorem \ref{thm:minCob}. Since GFG automata suffice for many purposes, for instance in a context where we want to test for inclusion, or compose the automaton with a game, this makes the class of GFG coB\"uchi automata particularly interesting. 

We examine later the case where GFG automata are not enough and we are aiming at building a DCA instead.

\subsection{Partial breakpoint construction}\label{sec:partbreak}
We generalise here the breakpoint construction from \cite{MH84}, in the same spirit as Section \ref{sec:partial}.

Let us first briefly recall the breakpoint construction. If $\A = (Q,\Sigma,\Delta,q_0,F)$ is an NCA, then a state of its determinized automaton is of the form $(X,Y)$, with $Y\subseteq X$. The powerset construction is performed on both sets, but the $Y$-component deletes rejecting states. The new transition function $\delta$ is defined as $\delta(X,Y)= (\Delta(X),\Delta(Y)\cap F)$, if $Y\neq\emptyset$. States with an empty second component are ``breakpoints'': they are rejecting, but they allow to reset the second component to the first one: $\delta(X,\emptyset)=(\Delta(X),\Delta(X))$.
The resulting deterministic run will accept if and only if the second component is eventually non-empty, witnessing the existence of an accepting run in $\A$.
\medskip

We now describe the incremental version of this construction.
For a parameter $k$, we want the $k$-breakpoint construction to be able to keep track of at most $k$ states simultaneously.


Given an NCA $\A = (Q,\Sigma,\Delta,q_0,F)$, we define the $k$-breakpoint construction of $\A$ as the NCA $\Ak = (Q', \Sigma,\Delta',(\{q_0\},\{q_0\}),F')$, with

$Q' = \{(X,Y) | X,Y \in \Qk \text{ and } Y\subseteq X\}$,
$$\Delta'((X,Y),a) := 
\begin{cases}
\{(\Delta(X,a),\Delta(X,a))\} & \text{if } Y=\emptyset\ \text{and}\ |\Delta(X,a)|\leq k\\
	\{(X',X') | \ X' \subseteq \Delta(X,a), |X'| = k\} & \text{if } Y=\emptyset\ \text{and}\ |\Delta(X,a)|> k\\
    \{(\Delta(X,a),\Delta(Y,a)\cap F)\} & \text{if }Y\neq\emptyset\ \text{and}\ |\Delta(X,a)|\leq k\\
\multicolumn{2}{l}{\{(X',X'\cap ( \Delta(Y,a)\cap F)) \ | \ X' \subseteq \Delta(X,a), |X'| = k\}~~ \text{otherwise}}
    
\end{cases}$$

$F' := \{(X,Y) \in Q'\ |\ Y \neq\emptyset\}$

That is, a run is accepting in $\Ak$ if it visits the states of the form $(X,\emptyset)$ finitely many times.

\begin{lem}\label{lem:numbreak}
The number of states of $\Ak$ is at most $\sum_{i=0}^k \binom{n}{i} 2^i$, which is in $O\big(\dfrac{(2n)^k}{k!}\big)$. 
\end{lem}
\begin{proof}
A state of $\Ak$ is of the form $(X,Y)$ with $|X|\leq k$ and $Y\subseteq X$. Therefore, there are at most $\sum_{i=0}^k \binom{n}{i} 2^i$ such states.
Since $\binom{n}{i}\leq \frac{n^i}{i!}$, we can bound the number of states by $\sum_{i=0}^k \frac{n^k}{k!}2^i\leq \frac{n^k}{k!}2^{k+1}=O\big(\frac{(2n)^k}{k!}\big)$
\end{proof}

\begin{lem}\label{lem:breakpoint_k} $L(A)=L(\Ak)$, and $\w(\A)\leq k \iff \Ak$ is GFG.
\end{lem}
\begin{proof}


%
%

\newcommand{\sigmaGFG}{\sigma_{GFG}}

%

First let us show that $L(\A)=L(\Ak)$.

Let $w\in L(\A)$, witnessed by an accepting run $\rho=q_0q_1q_2\dots$. The run $\rho$ can be used to resolve the nondeterminism of $\Ak$ while reading $w$, by choosing at each step $i$ as first component any set $X$ containing $q_i$. Since for all $i$ big enough, $q_i\in F$, after this point there is at most one empty second component in the run of $\Ak$, and therefore $w\in L(\Ak)$.

Conversely, let $w\in L(\Ak)$, witnessed by an accepting run $\rho=(X_0,Y_0)(X_1,Y_1)\dots$
We consider the DCA $\D$ obtained from $\A$ via the breakpoint construction \cite{MH84}. The states of $\D$ are of the form $(X,Y)$ with $X\supseteq Y$, and its acceptance condition is the same as the one of $\Ak$.
Let $\rho'=(X_0',Y_0')(X_1',Y_1')\dots$ be the run of $\D$ on $w$. By definition of $A_k$, for all $i\in\N$ we have $X_i\subseteq X_i'$ and $Y_i\subseteq Y_i'$. Since $Y_i$ is empty for finitely many $i$, it is also the case for $Y_i'$, and therefore $\rho'$ is accepting. We obtain $w\in L(\D)=L(\A)$.

Now we shall show $\w(\A)\leq k \Longrightarrow \Ak$ is GFG.


Assume that there is a winning strategy $\sigma_w:\Sigma^*\to \Qk$ for Player 0 in $\Gw(\A,k)$. We show that this induces a GFG strategy $\sigmaGFG:\Sigma^*\to Q'$ for $\Ak$. 
First, notice that without loss of generality, we can assume that for any $(u,a)\in \Sigma^*\times\Sigma$ such that $|\Delta(\sigma_w(u),a)|\leq k$, we have $\sigma_w(ua)=\Delta(\sigma_w(u),a)$. Indeed, it is always better for Player 0 to choose a set as big as possible. 

Using this assumption, the strategy $\sigmaGFG$ is naturally defined from $\sigma_w$ by relying on the first component, i.e. $\sigmaGFG(u):=(X',Y')$, where $X'=\sigma_w(u)$, and $Y'$ is forced by the transition table of $\Ak$. That is, if $Y\neq \emptyset$ we have $Y'=X'\cap \Delta(Y,a)\cap F$, and else ($Y=\emptyset$) we have $Y'=X'$.
We show that $\sigmaGFG$ is indeed a GFG strategy. Let $w$ be an infinite word in $L(\Ak)=L(\A)$. We must show that the run $\rho_w=(X_0,Y_0)(X_1,Y_1)(X_2,Y_2)\dots$ of $\Ak$ induced by $\sigmaGFG$ on $w$ is accepting.
Since $\sigma_w$ is a winning strategy in $\Gw(\A,k)$, there is an accepting run $\rho=q_0q_1q_2\dots$ of $\A$ such that for all $i\in\N$, $q_i\in X_i$. This means that there is $N\in\N$ such that for all $i\geq N, q_i\in F$. 
If for all $i\geq N$, $Y_i\neq\emptyset$, the run $\rho_w$ is accepting. Otherwise, let $M>N$ be such that $Y_M=\emptyset$. By definition of $\Ak$, we get $Y_{M+1}=X_{M+1}$. For $i\geq M+1$, we will always have $q_i\in Y_i$, by definition of $\Ak$, therefore $Y_i\neq \emptyset$. We can conclude that $\rho_w$ is accepting, and therefore $\Ak$ is GFG.

It remains to prove that $\Ak$ is GFG $\Longrightarrow \w(\A)\leq k$.

Let $\sigmaGFG:\Sigma^*\to Q'$ be a GFG strategy for $\Ak$. We build a strategy $\sigma_w:\Sigma^*\to \Qk$ for Player 0 in $\Gw(\A,k)$, witnessing that $\w(\A)\leq k$.

For all $u\in\Sigma^*$, we define $\sigma_w(u)$ to be the first component of $\sigmaGFG(u)$.
Let $w\in L(\A)$
= $L(\Ak)$, so the run $\rho=(X_0,Y_0)(X_1,Y_1)(X_2,Y_2)\dots$ induced by $\sigmaGFG$ on $w$ is accepting. We have to show that the play $\pi=X_0X_1X_2\dots$ is winning for Player 0 in $\Gw(\A,k)$, i.e. that there exists an accepting run $\rho_\pi=q_0q_1q_2\dots$ of $\A$ with $q_i\in X_i$ for all $i\in\N$. Assume no such run exists, i.e. all runs included in $\pi$ are rejecting.
Let $N\in\N$ be such that for all $i\geq N$, $Y_i\neq\emptyset$. Any run included in $\pi$ and starting in $q\in Y_N$ must encounter a non-accepting state. This means that there is $K>N$ such that between indices $N$ and $K$, every run included in $\pi$ contains a non-accepting state. By definition of $\Ak$, this implies there is $Y_i$ with $i>N$ such that $Y_i=\emptyset$. This contradicts the definition of $N$, therefore there must be an accepting run $\rho_\pi$ included in $\pi$.
\end{proof}



%
%
%
%
%

\subsection{Incremental construction of GFG NCA}

Supppose we are given an NCA $\A$, and we want to build an equivalent GFG automaton.

We can do the same as in Section \ref{sec:progNFA}: incrementally increase $k$ and test whether $\Ak$ is GFG, which is in PTIME by Theorem \ref{thm:GFG_P}. However in the coB\"uchi setting, the GFG automaton is not necessarily DBP, and can actually be more succinct than any deterministic automaton for the language (Theorem \ref{thm:succinct}).

If we are in a context where we are satisfied with a GFG automaton, such as synthesis or inclusion testing, this procedure can provide us one much more efficiently than determinisation.

Indeed, the example from Theorem \ref{thm:succinct} showing that GFG NCA are exponentially succinct compared to deterministic automata can be easily generalised to any width. For instance if our procedure is applied to the product of the GFG NCA of size $n$ from Theorem \ref{thm:succinct} with the one from Example \ref{ex:width2}, our construction will stop at the second step and generate a GFG automaton of quadratic size. This shows that the incremental construction for finding an equivalent GFG NCA can be very efficient compared to determinisation.

\subsection{Complexity of the width problem for NCAs}

As stated in Remark \ref{rem:omega}, directly computing the width of an NCA is EXPTIME-hard. 
The above construction together with Lemma \ref{lem:breakpoint_k} gives an EXPTIME algorithm solving the width problem for an input $\A,k$ with $\A$ an NCA: build $\A_k$, and test whether it is GFG in polynomial time. This shows that the EXPTIME-completeness of the width problem can be lifted to the coB\"uchi condition.

\begin{cor}
The width problem for NCAs is EXPTIME-complete.
\end{cor}

\subsection{Aiming for determinism}

In cases where a GFG automaton is not enough, and we want instead to build a DCA, we can test whether the current automaton is DBP instead of GFG in the incremental algorithm. If we find the automaton is DBP, we can remove the useless transitions, and obtain an equivalent DCA. This procedure will always stop, as in the worst case it will eventually reach the breakpoint construction, which directly builds a DCA.

Notice that the number of steps in this procedure corresponds to an alternative notion of width that can be called \emph{det-width}. The det-width of an automaton $\A$ is the least $k$ such that Player $0$ has a positional winning strategy in $\Gw(\A,k)$. Det-width always matches width on finite words by Theorem \ref{thm:loding}, but the notions diverge on infinite words.

This section studies the complexity of checking whether an NCA is DBP.
The next theorem shows that surprisingly, this is harder to check than being GFG for NCAs.

\begin{thm}\label{thm:npc} Given an NCA $\A$, it is NP-complete to check whether it is DBP.\end{thm}

\noindent We first show the hardness with the following lemma.

\begin{lem}\label{lem:NPh}
Checking whether an NCA is DBP is NP-hard.
\end{lem}
\begin{proof}
We prove this by reduction from the Hamiltonian Cycle problem on a directed graph, which is known to be NP-complete \cite{NPcompl}.

Recall that a Hamiltonian cycle is a cycle using each vertex of the graph exactly once.

Suppose we have a directed graph $G =  ([1,n], E)$ and we want to check whether it contains a Hamiltonian cycle. W.l.o.g. we can assume that the graph is strongly connected, otherwise the answer is trivially no. 

We construct an NCA $\A = (Q,\Sigma,\Delta,q_0,F)$, where $F$ is the set of accepting states, such that $\A$ is DBP if and only if $G$ has a Hamiltonian cycle.
The components of $\A$ are defined as follows: $Q := \bigcup_{i\in[1,n]} \{p_i,q_i,r_i\}$,
$\Sigma := \{a_1,a_2,\cdots,a_n,\#\}$,
$q_0 := p_1$,
$F :=\bigcup_{i\in[1,n]} \{p_i,q_i\}$, and finally $\Delta$ contains the following transitions, for all $i\in[1,n]$:
$$
\begin{array}{cccc}
    p_i \xrightarrow{a_i} q_i, &
    p_i \xrightarrow{a_j} r_i \text{ for all } j \neq i, &
    q_i \xrightarrow{\#} p_i, &
    \text{ and }r_i \xrightarrow{\#} p_k \text{ if } (i,k) \in E
\end{array}.$$

The only non-determinism in $\A$ occurs at the $r_i$ states when reading $\#$: we then have a choice between all the $p_k$ where $(i,k)\in E$.

We give an example for $G$ in Figure \ref{fig:GA}, where solid lines show the Hamiltonian cycle, together with a construction of $\A$ from $G$, where solid lines show a determinisation by pruning witnessing this Hamiltonian cycle. 
\bigskip

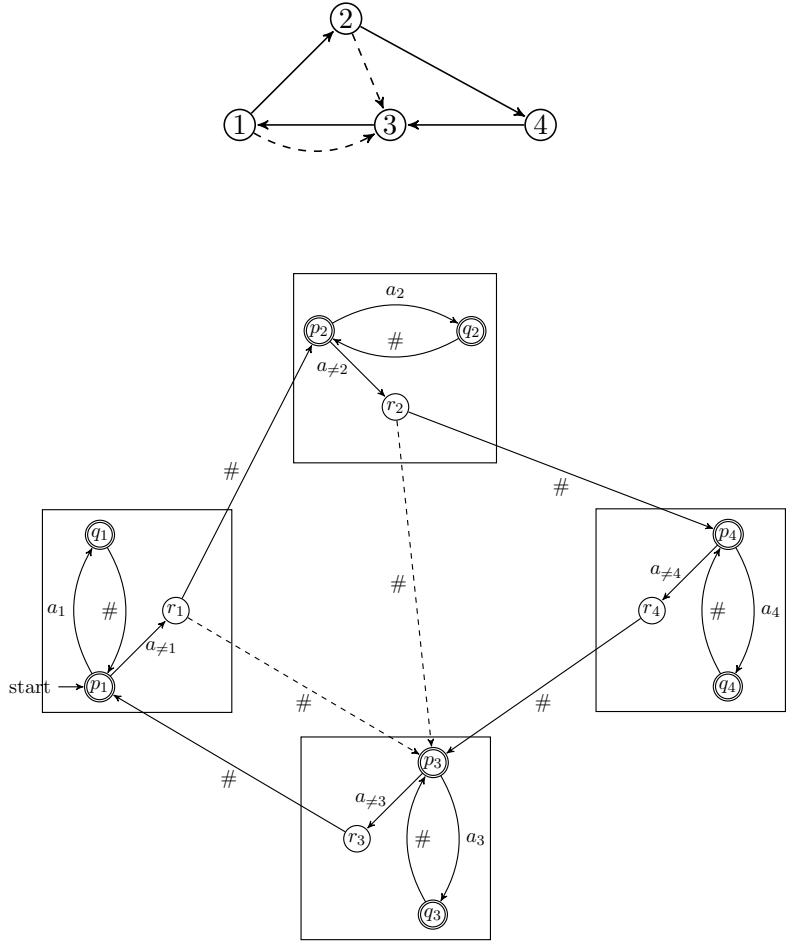
\begin{figure}[h]
\centering

\begin{tikzpicture}[->,>=stealth',shorten >=0.5pt,auto,node distance=2cm,
                    semithick]
  \tikzstyle{every state}=[fill=none,draw=black,text=black,minimum size=0pt,inner sep=1pt]

  \node[state] (A)                    {$1$};
  \node[state]         (B) [above right of=A] {$2$};
  \node[state]         (C) [right  of=A] {$3$};
  \node[state]         (D) [ right of=C] {$4$};

  \path (A) edge              node {} (B)
            edge            [dashed,bend right]  node {} (C)

        (B) edge              node {} (D)
	edge 	[dashed]	node{}(C)
        (D)
            edge              node {} (C)
        (C) edge node {} (A);
\end{tikzpicture}
\vspace{1.5cm}

\scalebox{.7}{
\begin{tikzpicture}[->,>=stealth',shorten >=0.5pt,auto,node distance=2cm,
                    semithick]
  \tikzstyle{every state}=[fill=none,draw=black,text=black,minimum size=0pt,inner sep=1pt]

  \node[state] (3)                    {$r_1$};
  \node[state,initial,accepting] (1)          [below left = 1.5cm of 3]          {$p_1$};
  \node[state,accepting] (2)          [above left = 1.5cm of 3]         {$q_1$};
  \node[state] (6)          [above right = 3.5cm and 3.8cm of 3]          {$r_2$};
  \node[state,accepting] (5)          [above right = 1.5cm of 6]          {$q_2$};
  \node[state,accepting] (4)          [ above left = 1.5cm of 6]          {$p_2$};
  \node[state,accepting] (7)          [below right = 2.5cm and 4.5cm of 3]          {$p_3$};
  \node[state] (9)          [below left = 1.5cm  of 7]          {$r_3$};
  \node[state,accepting] (8)          [below right = 1.5cm of 9]          {$q_3$};
  \node[state] (12)          [below right = 3.5cm and 4.5cm of 6]          {$r_4$};
  \node[state,accepting] (10)          [above right = 1.5cm of 12]          {$p_4$};
  \node[state,accepting] (11)          [below right = 1.5cm of 12]          {$q_4$};
  
  
\node[draw,inner xsep=8mm, inner ysep = 2mm, fit=(2) (1) (3) (1)] {};
\node[draw,inner xsep=2mm, inner ysep = 8mm, fit=(4) (6) (5) (4)] {};
\node[draw,inner xsep=8mm, inner ysep = 2mm, fit=(7) (8) (9) (8)] {};
\node[draw,inner xsep=8mm, inner ysep = 2mm, fit=(10) (11) (10) (12)] {};

  \path (1) edge        [bend left]      node {$a_1$} (2)
            edge              node [right] {$a_{\neq1}$} (3)
	(2) edge       [bend left]       node[left] {\#} (1)
	(3) edge            node[left] {\#} (4)
	(3) edge     [dashed]       node[below] {\#} (7)

	(4) edge        [bend left]      node {$a_2$} (5)
            edge              node [left=1pt] {$a_{\neq2}$} (6)
	(5) edge       [bend left]       node[above] {\#} (4)
	(6) edge      [dashed]        node[left] {\#} (7)
	(6) edge             node[below] {\#} (10)

	(7) edge        [bend left]      node [right]{$a_3$} (8)
            edge              node [left = .01cm] {$a_{\neq3}$} (9)
	(8) edge       [bend left]       node[right] {\#} (7)
	(9) edge              node[below] {\#} (1)

	(10) edge        [bend left]      node {$a_4$} (11)
            edge              node [left = .01cm] {$a_{\neq4}$} (12)
	(11) edge       [bend left]       node[right] {\#} (10)
	(12) edge              node[below] {\#} (7)
       ;
\end{tikzpicture}
}
\caption{Construction of NCA $\A$ (below) from $G$ (above)}\label{fig:GA}
\end{figure}

For each $i\in[1,n]$, we can think of the set of states $\{p_i,q_i,r_i\}$  as a cloud in $\A$ representing the vertex $i$ of the graph $G$.

Let $\Sigma' := \Sigma\setminus\{\#\}$, and $L = \bigcup\limits_{i=1}^n (\Sigma'\#)^*(a_i\#)^\omega$. First note that, provided $G$ is strongly connected, we have $L(\A) = L$. Indeed, for a run to be accepting by $\A$, it has to visit $r_i$ finitely many times for all $i$, i.e. after some point it has to loop between $p_i$ and $q_i$ for some fixed $i$, so the input word must be in $L$. This shows $L(\A)\subseteq L$. On the other hand, consider a word $w\in L$ of the form $u(a_i\#)^\omega$ with $u\in(\Sigma'\#)^*$. Then $\A$ will have a run on $u$ reaching some cloud $j$, and since the graph is strongly connected, the run can be extended to the cloud $i$ reading a word of $(a_i\#)^*$. From there, the automaton will read $(a_i\#)^\omega$  while looping between $p_i$ and $q_i$. We can build an accepting run of $\A$ on any word $w\in L$, so $L\subseteq L(\A)$.
\medskip

Now we shall prove that $\A$ is DBP if and only if G has a Hamiltonian cycle.

($\Rightarrow$) Suppose $\A$ is DBP, and let $\D$ be an equivalent DCA obtained from $\A$ by removing transitions.
Notice that since the only non-determinism in $\A$ is when reading a $\#$ letter in a $r_i$ state, by construction of $\A$, building $\D$ corresponds to choosing one out-edge for each vertex of $G$. This means it induces a set of disjoint cycles in $G$. We show that it actually is a unique Hamiltonian cycle. Indeed, assume that some vertex of $i$ is not reachable from $1$ in $G$. Equivalently, it means that some cloud $i$ is not reachable from $p_1$ in $\D$.
This implies that $(a_i\#)^\omega\notin L(\D)$, which contradicts $L(\D)=L(\A)=L$.
Therefore, $\D$ is strongly connected, and describes a Hamiltonian cycle in $G$.
\medskip

($\Leftarrow$) Conversely, if $G$ has an Hamiltonian cycle $\pi$ , we can build the automaton $\D$ accordingly, by setting for all $i\in[1,n]$, $\Delta_\D(r_i,\sharp)=\{p_j\}$ where $j$ is the successor of $i$ in $\pi$.
Since $\D$ is strongly connected, it still recognises $L$, and since it is deterministic it is a witness that $\A$ is DBP.
\medskip

This completes the proof of the fact that $\A$ is DBP if and only if $G$ has a Hamiltonian cycle.
Since this is a polynomial time reduction from Hamiltonian Cycle to the DBP property of an NCA, we showed
that checking whether an NCA is DBP is NP-hard.

Note that we used $n+1$ letters here, but it is straightfoward to re-encode this reduction using only two letters. Therefore, the problem is NP-hard even on a two-letter alphabet. It is trivially in PTIME on a one-letter alphabet, as there is a unique infinite word.
\end{proof}

The second part of Theorem \ref{thm:npc} is given by the following lemma.

\begin{lem}\label{lem:NPmem}
Checking whether an NCA is DBP is in NP.
\end{lem}

\begin{proof}
Suppose an NCA $\A$ is given. We want to check whether it is DBP.
We do this via the following NP algorithm.

\begin{itemize}
\item
Nondeterministically prune transitions of $\A$ to get a deterministic automaton $\D$.
\item
Check whether $L(\A)\subseteq L(\D)$. For that, we check if $L(\A)\cap\overline{L(\D)}=\emptyset$ 
\end{itemize}

\noindent
The second step of the algorithm can be done in polynomial time, since it amounts to finding an accepting lasso in $\A\times \overline{\D}$, where $\overline{\D}$ is a B\"uchi automaton obtained by dualizing the acceptance condition of $\D$. An accepting lasso is a path of the form $q_0\srelT{~u~}{} p\srelT{~v~}{} p$, witnessing that the word $uv^\omega$ is accepted, i.e. in this case the loop should visit only accepting states from the NCA $\A$, and at least one B\"uchi state from the DBA $\overline{\D}$. Finding such a lasso is actually in NL.
%
Therefore, the above algorithm is in NP, and its correctness follows from the fact that $L(\D)\subseteq L(\A)$ is always true, as any run of $\D$ is in particular a run of $\A$.
\end{proof}

\section{B\"uchi Case}\label{sec:buchi}

Non-deterministic B\"uchi automata (NBA) correspond to the general case of non-deterministic $\omega$-automata, as they allow to recognise any $\omega$-regular language, and are easily computable from non-deterministic automata with more general acceptance conditions.

We will briefly describe the generalisation of previous constructions here, and explain what is the main open problem remaining to solve in order to obtain a satisfying generalisation.
We take Safra's construction \cite{Safra88} as the canonical determinisation for B\"uchi automata. Safra's construction outputs a Rabin automaton.

The idea behind the previous partial determinisation construction can be naturally adapted to Safra's: it suffices to restrict the image of the Safra tree labellings to sets of states of size at most $k$.
The bottleneck of the incremental determinisation is then to test whether a Rabin automaton is GFG (or DBP). For DBP checking, the same proof as Theorem \ref{thm:npc} shows that it is NP-complete. However for GFG checking, the complexity is widely open. It is only currently known to be in P for the particular cases of coB\"uchi \cite{KS15} and B\"uchi \cite{BK18} conditions.
A lower bound for GFG checking with acceptance condition $C$ (for instance $C$ being Parity or Rabin) is the complexity of solving games with winning condition $C$ \cite{KS15}, known to be in QuasiP for parity \cite{CaludeJKLS17} and NP-complete for Rabin \cite{EJ88}. In both cases, the complexity of solving those games is in P if the acceptance condition $C$ is fixed (for instance $[i,j]$-parity for fixed $i$ and $j$). On the other hand, the best known upper bound for the checking the GFG property is EXPTIME \cite{KS15}, even for a fixed acceptance condition $C$ such as Parity with $3$ ranks.
Finding an efficient algorithm for GFG checking of Rabin (or Parity) automata would therefore be of great interest for this incremental procedure, and would allow to efficiently build GFG automata from NBA.

\subsection{Safra Construction}
\newcommand{\green}{\mathit{Green}}
\newcommand{\white}{\mathit{White}}

Let $\A$ be an NBA $(Q, \Sigma, \Delta, q_0, F)$ where $F$ is the set of accepting states, and let $n=|Q|$. Safra construction produces an equivalent deterministic Rabin automaton $\D = (Q', \Sigma, \Delta', q_0', F')$ with $2^{O(n \log n)}$ many states \cite{Safra88,M12}.

We recall here the construction, in order to adapt it to an incremental construction computing the width and producing a GFG automaton.

Each state in $Q'$ is a tuple $(T, \sigma, \chi, \lambda)$ where
\begin{itemize}
\item
$T = (V, v_r, \cl)$ is a tree where $V$ is the set of nodes, $v_r\in V$ is the root and $\cl$ (for \emph{children list}) is a function $V\to V^*$ mapping each node to the ordered list of its children, from left to right.
\item
$\sigma : V \rightarrow 2^Q$ maps each node to a set of states, such that
\begin{enumerate}
\item
For each $v\in V$, if $\cl(v)=v_1\dots v_k$, then $\sigma(v_1)\cup\dots\sigma(v_k)\subsetneq \sigma(v)$.
\item
If $v$ and $v'$ are two nodes such that none of them is ancestor of the other then $\sigma(v)$ and $\sigma(v')$ are disjoint.
\item
If $\sigma(v) = \emptyset$, then 
 $v$ must be the root node $v_r$.
\end{enumerate}
\noindent Note that these conditions imply that $|V|\leq n$.
\item
$\chi : V \rightarrow \{\green,\white\}$ assigns a colour to each node.
\item
$\lambda : V \rightarrow \{l_1, l_2, \cdots , l_{2n}\}$ associates a label to each node in V.
\end{itemize}
\noindent The initial state $q_0'$ is $(T_0,\sigma_0,\chi_0,\lambda_0)$, where $T_0$ contains only the root $v_r$, $\sigma_0(v_r) = \{q_0\}$, $\chi_0(v_r) = \white$, $\lambda_0(v_r) = l_1$.

Now we define $\Delta'$. The state $(T, \sigma, \chi, \lambda)$, reading $a\in \Sigma$,  is moved to a state $(T', \sigma', \chi', \lambda')$ as follows :
\begin{enumerate}
\item[(i)] Let $T = (V,v_r,\cl)$. First expand the tree $T$ to $T_1 = (V_1, v_r, \cl_1)$ as follows: for each node $v\in V$, if $\sigma(v) \cap F \neq \emptyset$, then add a node $v'$ such that $v'$ is the right-most child of $v$ in $T_1$.
\item[(ii)] Extend $\sigma$ to $\sigma_1$ as follows:
for all $v\in V\cap V_1,$ set $\sigma_1(v) = \sigma(v)$.
And for each new node $v\in V_1\setminus V$, let $p$ be the parent node of $V$ in $T_1$, set $\sigma_1(v) = \sigma(p)\cap F$.

Extend $\lambda$ to $\lambda_1$ as follows: for all $v\in V\cap V_1$, $\lambda_1(v) = \lambda(v)$. And for each new node $v\in V_1 \setminus V$, fix a new label to $v$ which was not used in $V$.
We can always find such a label since there are $2n$ labels whereas $|V|\leq n$ and each node in $V$ generates atmost one new node in $V_1$.
\item[(iii)] For each node $v\in V_1$, apply the subset construction locally, i.e. define $\sigma_1':V_1\rightarrow2^Q$ such that $\sigma_1'(v) = \Delta(\sigma_1(v),a)$. As before, $\Delta(X,a)$ denotes the set $\{q'\mid \ \exists q \in X, \ q'\in\Delta(q,a)\}$).
\end{enumerate}
Now we modify $T_1$ and $\sigma_1'$ so that the structure satisfies the conditions specified for the states of $\D$ as follows:
\begin{enumerate}
\item[(iv)]

For every node $v \in V$, if there is some $s \in \sigma_1'(v)$ and also $s\in\sigma_1'(v')$ for some node $v'$ 
such that $v$ and $v'$ have a common ancestor, say $u$, and $v$ is in the subtree rooted at a child $u_1$ of $u$ and $v'$ is in the subtree rooted at a child $u_2$ of $u$ where $u_1$ is to the left of $u_2$, then remove $s$ from $\sigma_1'(v')$. This corresponds to retaining only the ``oldest'' copy of each active state in the simulation.
\item[(v)]
Remove all nodes $v$ such that $\sigma_1'(v) = \emptyset$ and $v$ is not the root.
\item[(vi)]
For each node $v$, if 
 the union of the sets associated with the children of $v$ is equal to $\sigma(v)$, then
 remove all the children of $v$ and make $\chi(v) = Green$. And for all other nodes, set $\chi(v) = White$.
\end{enumerate}
Let the set of remaining nodes be $V'$. And $\sigma', \lambda', \chi'$ are retained from the nodes in $V'$.

If $T$ is a tree, we denote by $V(T)$ its set of nodes.

The Rabin acceptance condition is given by $\{(G_1,B_1), (G_2,B_2), \cdots, (G_{2n},B_{2n})\}$ where $G_i$ are the good states and $B_i$ are the bad states and they are defined as follows:
\begin{itemize}
\item
$G_i = \{(T,\sigma,\chi,\lambda)\in Q' \ | \ \exists v\in V(T) : \lambda(v) = l_i \text{ and } \chi(v) = Green\}$
\item
$B_i = \{(T,\sigma,\chi,\lambda)\in Q' \ | \ \forall v\in V(T) : \lambda(v) \neq l_i\}$
\end{itemize}
That is to say, a run is accepting if there is some $i$ such that states from $G_i$ appear infinitely often while states from $B_i$ appear only finitely often.

\medskip
The Safra construction is an efficient way to compress information about all possible runs of $\A$. Indeed,  on a word $w=a_0a_1a_2\dots$, the set of runs on $w$ can be described by an infinite Direct Acyclic Graph (DAG) called the run-DAG. This run-DAG has vertices $Q\times\N$, and edges $\{((p,i),(q,i+1)) \mid (p,a_i,q)\in\Delta\}$. We have that $w\in L(\A)$ if and only if there is a path in this run-DAG starting in $(q_0,0)$ and visiting infinitely many states from $F$. Safra trees store relevant information about prefixes of this DAG of the form $Q\times[0,i]$, and the acceptance condition of $\D$ is designed to characterize whether the run-DAG of $\A$ contains a B\"uchi accepting run.

\begin{thm}\cite{Safra88}\label{thm:Safra}
The deterministic Rabin automaton $\D$ built via the Safra construction is equivalent to $\A$.
\end{thm}

\subsection{Incremental Safra Construction} 

We can extend the concept of $k$-subset construction or $k$-breakpoint construction for NFA and NCA respectively to the Safra construction. We describe below the $k$-Safra construction, for a parameter $k\leq n$.

Here we restrict $\sigma(v_r)$, $v_r$ being the root, to sets of size at most $k$. Since all other nodes in Safra trees are labelled by subsets of $\sigma(v_r)$, this implies that the labelling of all the nodes in all Safra trees have size at most $k$.
We define the construction formally as follows:

Given an NBA $\A = (Q, \Sigma, \Delta, q_0, F)$. Define the NRA $\A_k = (Q', \Sigma, \Delta', q_0', F')$ where each state in $Q'$ is $(T,\sigma,\chi,\lambda)$ such that
\begin{itemize}
\item
$T = (V,v_r,\cl)$ as in the original construction.
\item
$\sigma : V \rightarrow 2^Q$ satisfying the same properties as before. Additionally we also have the condition that $|\sigma(v)|\leq k$ for all $v\in V$.
\item
$\chi$ and $\lambda$ are also defined as before.
\end{itemize}

\noindent Now we define $\Delta'$.
All the steps remain unchanged except step (iii), which is replaced by the following step, nondeterministically choosing a subset of size $k$ for the root, and propagating it down in the tree:

$\sigma_1'(v_r) := 
\begin{cases}
    \{\Delta(\sigma_1(v_r),a)\}						   & \text{if }|\Delta(\sigma_1(v_r),a)| \leq k\\
    \{X' \mid X'\subseteq \Delta(\sigma_1(v_r),a),|X'| = k\}				   & \text{otherwise}
\end{cases}$

\noindent And for every other node $v \neq v_r$, $\sigma_1'(v) := (\Delta(\sigma_1(v),a))\cap\sigma_1'(v_r)$.

This corresponds to extending the run-DAG by at most $k$ nodes (the label of the root) at each step, i.e. non-deterministically building a subset of the run-DAG of width at most $k$.

Initial states and acceptance condition are defined as before.

We can now state that this $k$-Safra construction characterizes the notion of width in the same way the previous incremental constructions did:
\begin{lem}
$L(\A) = L(\A_k)$ and $\w(\A)\leq k\iff \A_k$ is GFG.
\end{lem}
\begin{proof}
Suppose that $w\in L(\A)$, witnessed by the run $\rho = q_0q_1q_2\cdots$. Then in $\A_k$, at step $i$, choose a set $X$ of size at most $k$ containing $q_i$ as $\sigma(v_r)$ where $v_r$ is the root of the tree and for all other nodes, take its intersection with the label of the root. Since the run-DAG that is built contains a B\"uchi accepting run, the correctness of the original Safra construction ensures that this run is accepting, so $w\in L(\A_k)$.

Conversely, let $w\in L(\A_k)$. This means that the run-DAG of width $k$ guessed by the automaton contains a B\"uchi accepting run of $\A$ on $w$, so $w\in L(\A)$.
\medskip

Now suppose that $\w(\A)\leq k$ and let $\sigma_w$ be a winning strategy for Player $0$ in $\Gg_w(\A,k)$. The only non-determinism in $\A_k$ consists in choosing subsets of size $k$ for the label of the root of Safra trees. So $\sigma_w$ naturally induces a GFG  strategy $\sigma_{GFG}$ for $\A_k$, by choosing as root labelling the subset given by $\sigma_w$. Again, the correctness of the Safra construction implies that this GFG strategy is correct, as the run-DAG will contain a B\"uchi accepting run, by definition of the winning condition in $\Gw(\A,k)$.

Conversely, let $\sigma_{GFG}$ be a GFG strategy for $\A_k$. It induces a winning strategy $\sigma_w$ for Player $0$ in $\Gg_w(\A,k)$ which follows the labellings of the root nodes in the run induced by $\sigma_{GFG}$.
\end{proof}

Therefore, we can design a similar incremental approach as in the case of NFA or NCA, and find the minimum $k$ for which the $k$-Safra construction is GFG, or DBP.

\subsection{Complexity of GFG and DBP checking}

As mentioned in the beginning of Section \ref{sec:buchi}, the complexity of checking whether an automaton is GFG in the general case of Rabin or Parity automata is widely open. It is only known to be in P for coB\"uchi condition \cite{KS15} and for B\"uchi condition \cite{BK18}, while the upper bound is EXPTIME for higher conditions, such as parity condition with $3$ ranks.
This means that in the current state of knowledge, the incremental Safra construction combined with GFG checking only yields a doubly exponential upper bound for the width problem of B\"uchi automata. Recall that the lower bound for the width problem is EXPTIME also in this case, by Remark \ref{rem:omega}.
Obtaining a better upper bound via a different route is an interesting problem, that we leave open.

On the other hand, if we aim at obtaining a deterministic automaton, the picture is simplified and matches the coB\"uchi case. Indeed, for DBP checking, we can generalise Theorem \ref{thm:npc} stating NP-completeness of DBP checking for coB\"uchi automata:

\begin{thm}
Checking whether a Rabin (or Streett, Muller) automata is DBP is NP-complete.
\end{thm}
\begin{proof}
Since coB\"uchi condition is a particular case of Rabin condition, NP-hardness follows from Lemma \ref{lem:NPh}.

We now show membership in NP, in the same way as in Lemma \ref{lem:NPmem}.

As before, we non-deterministically choose a set of edges to remove in order to obtain a deterministic Rabin automaton $\D$.
Then, it remains to decide emptiness of the automaton $\A\times\overline{\D}$, where $\A$ is a non-deterministic Rabin automaton and $\D$ is a deterministic Streett automaton (Streett is the condition dual to Rabin). Again, emptiness of this automaton amounts to guessing an accepting lasso $q_0\srelT{u}{} p\srelT{v}{} p$. There is no direct NL algorithm to guess such a lasso as in the proof of Theorem \ref{thm:npc}, since in order to verify the Streett condition of $\overline{\D}$, all Streett pairs must be checked. However, this can be done in NP: guess the lasso, guess the Rabin pair to witness acceptance of $\A$ (or verify them all), and verify that all Streett pairs of $\overline{\D}$ are verified by the loop: for each Streett pair $(E,F)$, either the loop contains a state from $E$ or no state from $F$.
This procedure can be generalised to Muller acceptance condition for $\A$: it is always in P to decide whether a particular loop verifies a boolean combination of Muller acceptance conditions.
Overall, this yields an NP algorithm for checking whether a Rabin (or Streett, Muller) automaton is DBP.
\end{proof}

\section{Conclusion}

The width measure can be thought of as a measure of non-determinism in automata, that is essentially orthogonal to ambiguity. It also bounds the number of steps in our incremental determinisation procedures. Therefore, if we know that the width is small, we can obtain a deterministic or GFG automaton without having to go through a full determinization construction as intermediary step. The EXPTIME-completeness of the width problem shows that there is essentially no shortcut that would allow to jump directly to the good level of the incremental construction by computing the optimal width without performing the incremental construction. A dichotomic approach could still present some advantages in practice by saving a few computations in the process of finding the width, but the EXPTIME barrier is a strict theoretical limit.

The cases of finite words and coB\"uchi condition are especially well-suited for this approach. Indeed, these conditions allow polynomial time checking of the GFG property \cite{LR13,KS15}, and polynomial time minimization of the resulting GFG automaton \cite{AK19}. The NP-completeness of DBP checking for coB\"uchi automata is another reason to aim for a GFG automaton when determinism is not strictly required.

For B\"uchi automata, we need to build increasingly complex Rabin automata $\A_k$ via the $k$-Safra construction. The complexity of checking whether such automata are GFG is less well-understood. As future work, it is pertinent to either search for a special GFG checking procedure well-suited to these $\A_k$, or on the contrary show that checking the GFG property for these particular $\A_k$ is as hard as for general automata. In any case, this problem provides additional motivation to pinpoint the complexity of GFG checking in general.

\bigskip

\noindent\textbf{Acknowledgements.} We thank the anonymous reviewers for their detailed feedback, and Nathalie Bertrand for helpful discussions.

\bibliographystyle{plainurl}
\bibliography{width}
\end{document}